\documentclass[reqno,centertags,12pt]{amsart}

\usepackage[utf8]{inputenc}
\usepackage[T1]{fontenc}
\usepackage[american]{babel}
\usepackage[babel=true,expansion=alltext,protrusion=alltext-nott,final]{microtype}
\usepackage{times}
\usepackage{dsfont}
\usepackage{mathtools}
\usepackage{bbm}
\usepackage[hidelinks]{hyperref}
\usepackage{tikz}
\usetikzlibrary{angles, quotes,decorations.pathreplacing}

\usepackage{xcolor}

\hypersetup{pdftitle={}, pdfauthor={D. Hundertmark}, pdfsubject={35P15; 35J10, 81Q10}, pdfkeywords={}}

\usepackage[utf8]{inputenc}
\usepackage[sb]{libertine}
\usepackage[varqu,varl]{inconsolata}
\usepackage[libertine]{newtxmath}
\useosf
\usepackage[T1]{fontenc}
\usepackage{textcomp}

\usepackage{amsthm}
\usepackage{enumitem,multirow}
\usepackage{mhchem}

\usepackage[margin={3cm},marginparwidth={2.5cm}]{geometry}

\newcommand{\R}{{\mathbb{R}}}

\newcommand{\abs}[1]{\left| #1 \right|}
\newcommand{\norm}[1]{\left\Vert #1 \right\Vert}
\allowdisplaybreaks
\DeclareMathOperator{\supp}{supp}

\newtheorem{theorem}{Theorem}[section]

\newtheorem{lemma}[theorem]{Lemma}

\theoremstyle{definition}

\newtheorem{remark}[theorem]{Remark}
\newtheorem{remarks}[theorem]{Remarks}

\numberwithin{theorem}{section}
\numberwithin{equation}{section}

\DeclareMathOperator {\ess}{ess}
\DeclareMathOperator {\loc}{loc}
\DeclareMathOperator {\disc}{disc}

\newcounter{smalllist}

\usepackage{color}
\thanks{\copyright 2024 by the authors. Faithful reproduction of this article,
       in its entirety, by any means is permitted for non--commercial purposes}

\date{\today, version \jobname }

\title[Absence of the Efimov Effect for Confined Particles]{Why a System of three bosons on separate lines can not exhibit the Confinement Induced Efimov effect}
\author{Dirk Hundertmark, Marvin R. Schulz, Semjon Vugalter}

\begin{document}

\begin{abstract}
We study a system of three bosons interacting with short--range potentials which can move along three different lines. Two of these lines are parallel to each other within one plane. The third line is constrained to a plane perpendicular to the first one. Recently it was predicted in physics literature \cite{N:2017} that such a system exhibits the so--called confinement induced Efimov effect. We prove that this prediction is not correct by showing that this system has at most finitely many bound--states.
\end{abstract}

\maketitle
\setcounter{tocdepth}{1}
{\hypersetup{linkcolor=black} 
\tableofcontents }

\section{Introduction} \label{sec:intro}
\subsection{The physical system} \label{subsec:model} 
We study a system consisting of three quantum particles (bosons) with \emph{short--range} interactions confined to move on separate lines within \(\mathbb{R}^3\), see Figure \ref{fig:001}  in section \ref{main_result:section}. 
Two of these lines, \(L_2\) and \(L_3\), are parallel to each other within a plane \(P\), while the first line, \(L_1\), is constrained to a plane perpendicular to \(P\) which does not intersect \(L_2\) or \(L_3\). That is, the intersection of the plane on which \(L_1\) lies with the plane \(P\) forms a line parallel to  \(L_2\) and \(L_3\). The line \(L_1\) intersects \(P\) at an angle \(\zeta \in (0,\pi/2]\). 
Without loss of generality, we can fix the point of intersection 
of the line \(L_1\) and the plane \(P\) as the origin. 

Recently the 
physicists Nishida and Tan predicted that this system might exhibit the so--called confinement induced Efimov effect. This in particular means, that it should have an infinite number of negative eigenvalues, if two--particle subsystems do not have bound states, but have resonances at the bottom of the essential spectrum. the possible existence of Efimov type effects in such \emph{}{geometrically constrained quantum systems}, see \cite{NT:2008},\cite{NT:11}). 

\smallskip 

Our main result, see Theorem \ref{main_thrm}, shows that this prediction is \emph{not correct}. We prove that the  
geometrically constrained three--particle system discussed above 
can have at most finitely many bound--states for a large class of short--range potentials.

\smallskip 

\subsection{The Efimov effect.} It is well known that an one-particle Schr\"odinger operator $-\Delta + V(x)$ on $L^2(\R^d)$ with relatively bounded potential decaying faster than $\abs{x}^{-2-\delta}$, $\delta>0$ may have only a finite number of eingenvalues. It was proven by Zhislin \cite{GZ:1974} and Yafaev \cite{J:1974} using two different methods that $N$-particle Schrödinger operators, under the same conditions on the potentials, posess only a finite number of eigenvalues if at least one of the subsystems has a bound state below zero.

It was very surprising, when physicist Efimov found in 1970 \cite{E:1970} that a system of three particles in $\R^3$ with short--range pairwise interactions may have an \emph{infinite number} of eigenvalues when the two--particle subsystems do not have bound--states, but have resonances at the bottom of the spectrum. 

Beyond the infinite accumulation of bound--states, the Efimov effect exhibits several remarkable properties, 
one of the most significant being its universality. 
This means that the discrete spectrum's asymptotic behavior remains the same, regardless of the microscopic 
specifics of the underlying pair--potentials. 
In particular, the number of bound--states \( N(E) \) below \( E < 0 \) satisfies the \emph{universal} asymptotic behavior 
\begin{equation} \label{scaling}
    \lim_{E \to 0^-} N(E) = C_0 \abs{\ln(\abs{E})} + o(1) 
\end{equation}
for some constant \( C_0 > 0 \), which depends solely on the particle masses and not on the interaction potentials. 

Both in mathematics and physics, it became a highly recognized challenge to study the Efimov effect. 
After Efimov's initial description the first rigorous mathematical proof was provided by Yafaev in 1970 \cite{J:1974}, followed by a variational proof by Ovchinnikov and Sigal in 1979 \cite{OS:79} and Tamura in \cite{T:1991}. 
The asymptotic behavior of the number of bound--states, which was already predicted by Efimov, was later confirmed mathematically by Sobolev in 1993 \cite{S:1993}. Until the end of the 1990s, several significant physical and mathematical findings had emerged on this topic (see, e.g., \cite{T:1993}, \cite{KS:79},  \cite{P:1995}, \cite{P:1996}, \cite{VZ:1982} and \cite{VZ:1984})

Despite its universality property, the Efimov effect is 
an exceptionally rare phenomenon, primarily due to the necessary presence of virtual levels in two--particle subsystems.
In experiments, it is difficult to create conditions where two--particle subsystems have zero--energy resonances. Moreover, the Efimov bound--states have a large size and are very weakly bound. This makes the Efimov effect exceedingly challenging to observe.

However, technological advancements in the 1990s, such as improved laser cooling techniques, enabled the study of resonant systems through the application of magnetic fields and so--called \emph{Feshbach resonances} (see, e.g., \cite{T:1993}, \cite{I:Exp:1998}, and \cite{C:Exp:1998}). The first experimental observation of the Efimov effect was achieved in 2002 in an ultracold gas of cesium atoms, which was published in 2006 \cite{K:Exp:2006}. Later, in experiments with potassium atoms, two consecutive Efimov states have been observed \cite{Z:2009}, obtaining data consistent with the universal scaling property in \eqref{scaling}. By the late 2000s, evidence of the Efimov effect was found for various other particles, including cases beyond systems of three identical bosons (see, e.g., \cite{G:2009}, \cite{G:Li7:2009}, \cite{P:2009}, \cite{B:2009} and \cite{X:2020}). For a more detailed review of the experimental results on the Efimov effect see \cite{F:2010} and the references therein. The experimental verification of the Efimov effect generated renewed and significant academic interest. For a comprehensive review, see \cite{N:2017} (published in 2017), which includes over 400 references, most of them from after 2009.

Recently, both experimental and theoretical studies have explored the existence of effects similar to the Efimov effect. A natural question is: does the Efimov effect extend to \(N\)--particle systems for \(N>3\) when the \((N-1)\)--particle subsystems possess a virtual level? It is known that in systems with \(N\geq 4\) bosons in three dimensions, the effect is absent \cite{AG:1973}, \cite{arxive:BBV}, \cite{1D:13}. 

Another question is whether the effect can exist in spatial dimensions other than three, which occur, for example in configurations involving graphene or by confinement of particles via optical lattices. In systems of $N$ bosons, the absence of the Efimov effect in dimension one has been established in \cite{arxive:BBV}. In the same work, it was proved that for $N$ two--dimensional bosons the Effimov effect is absent, except in the case $N=4$. Physicists predict that for $N=4$ the Efimove effect exists only if the system interacts solely via three-particle forces. Mathematically this is still an open problem.

For $N=3$ in dimensions greater than five, virtual levels correspond to bound--states of two--particle subsystems, resulting in the non--existence of the Efimov effect. The situation in dimension four is more complex since virtual levels in this case are resonances but not bound--states, however, the decay rate is so high, that the resonance barely misses to be $L^2$. The non--existence of the Efimov effect for three bosons in dimension four was demonstrated by the use of so--called \emph{Faddeev equations} in \cite{BB:2019}. This completes the picture for the existence or non--existence of the Efimov effect for (bosonic) three--particle systems in all dimensions.

Recent advancements in experiments with ultracold gases have enabled the confinement of particles to lower--dimensional subspaces using strong optical lattices (see, e.g., \cite{G:2001}, \cite{T:2008}). 
This development enables the study of systems with \emph{mixed dimensionality}, where different species of particles are confined to distinct subspaces of dimension less than three. 
The physicists Nishida and Tan discussed the possible existence of Efimov type effects in such geometrically constrained systems (see \cite{NT:2008},\cite{NT:11}). 
In \cite{NT:2009}, they examined the possibility of this effect occurring in a mixture of \ce{^{40}K} and \ce{^{6}Li} isotopes, where the conventional Efimov effect is known to be absent. 
They argued that confining \ce{^{40}K} particles to a one--dimensional subspace is a promising system for the so--called \emph{confinement induced Efimov effect}. However, a rigorous mathematical description of these scenarios remains an open question. 
\begin{figure}[!ht]
\centering
\includegraphics[width=0.9\textwidth]{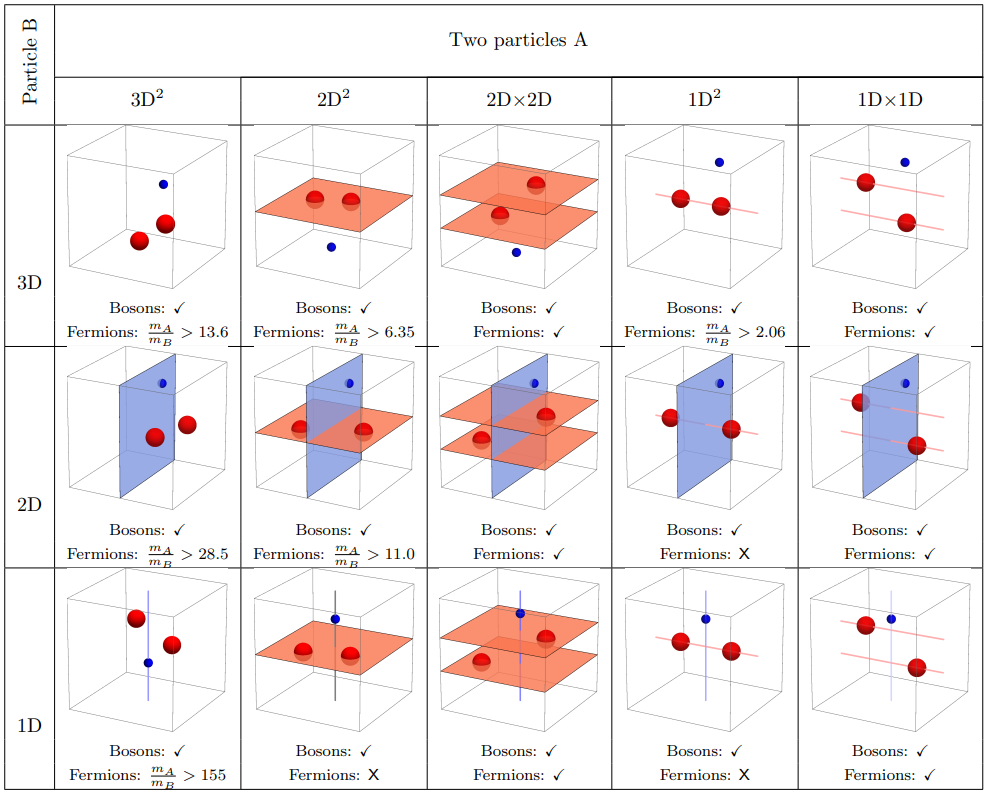}
\caption{Predictions on the confinement induced Effimov effect. Graphic taken from \cite[page 44, Table 1]{N:2017}. For each case, it is indicated whether the Efimov effect is predicted to occur $(\checkmark)$ or not $(x)$.} \label{fig:005}
\end{figure}
In \cite{N:2017} the predictions on the Efimov effect among various configurations of three particles with mixed dimensionality have been summarized. We present the table \cite[page 44, Table 1]{N:2017} and these predictions in Figure \ref{fig:005} below. In this work, we examine one of the cases. Namely the case \mbox{(1D-1D x 1D)} on the bottom right.  

We prove that, contrary to the predictions made, this system can support only a finite number of bound--states, even if there are virtual levels in the two--particle subsystems. 

This configuration is particularly intriguing because it is a truly mixed dimensional system, consisting of one one--dimensional and two two--dimensional subsystems.

Our approach is based on methods developed by Vugalter and Zhislin (see, e.g., \cite{GZ:1974}, \cite{VZ:1982}, \cite{VZ:1984} and \cite{V:1992}), which were recently applied in \cite{BBV:2022} to establish the absence of the Efimov effect in unconstrained $N$--particle systems in dimensions one and two. As usual, an important part of the work is the study of the decay properties of resonances which may occur at the bottom of the spectrum of subsystems. To the best of our knowledge, the decay of resonances has only been studied in the case of rotational symmetric potentials. However, for this mixed dimensional system, the subsystems are not necessarily rotational invariant which tremendously complicates the analysis of the decay properties of zero--energy solutions. Although this solutions are not functions in $L^2(\R^2)$ using a modification of techniques from \cite{HJL-helium}, \cite{HJL-proceedings},  \cite{arxive:BBV}, \cite{BBV:2022} and \cite{BHHV:2022},
which extend the method of \cite{Agmon:Lectures}, we show that the projection of these solutions onto the subspace orthogonal to radially symmetric functions are in $L^2(\R^2)$.

The paper is structured as follows. In Section \ref{main_result:section}, we give main definitions and state our main result and addresses the (lack of) symmetries within the two--particle subsystems. In Section \ref{proof_of_main_thrm}, we state four lemmas and show how they 
lead to a proof of our main theorem regarding the finiteness of the discrete spectrum.

As preparation for proving these lemmas, Section \ref{preliminiaries} examines the decay properties of the zero--energy solutions within a 
given symmetry subspace. It turns out that the part of a resonance in angular momentum subspaces different from zero angular momentum 
has faster decay than the part in the zero angular momentum subspace (the s channel in physics language).  
We finish the proof of  the main theorem by 
proving the remaining lemmas from Section \ref{proof_of_main_thrm} 
in Section \ref{proofs}.

\smallskip\noindent
\textbf{Acknowledgements:} 
This research has been partially funded by the Deutsche Forschungsgemeinschaft (DFG, German Research Foundation) – Project-ID 258734477 – SFB 1173. 
\section{Definitions and Main Result} \label{main_result:section}
Let $y_i \in \R$ be the distance of the $i$-th particle from the origin along the line $L_i$, and let $\mathbf{r}_i \in \mathbb{R}^3$ be the three--dimensional position vector of this particle. Then
\begin{equation}
    \mathbf{r}_1 = \left(
\begin{array}{c}
y_1 \cos(\zeta)\\
0\\
y_1 \sin(\zeta)\\
\end{array}
\right), \,
\mathbf{r}_2 = \left(
\begin{array}{c}
y_2\\
a_2\\
0\\
\end{array}
\right), \,
\mathbf{r}_3 = \left(
\begin{array}{c}
y_3\\
a_3\\
0\\
\end{array}
\right),
\end{equation}
where $a_2,a_3 \in \R$ denote the distances between the lines as indicated in Figure \ref{fig:001}. 

Denote by $\mathbf{r}_{ij}= r_i-r_j$ the distance between 
the particles $i$ and $j$.  
The Schr\"odinger operator of the system, expressed in this coordinate system, is given by
\begin{equation}\label{3-particle-ham}
    H =  - \sum_{i=1}^3 \frac{1}{m_i}\frac{\partial^{2}}{\partial y_i^{2}} +\sum_{\alpha\in I } V_{\alpha}(\abs{\mathbf{r}_{\alpha}})
\end{equation}
where \( V_{\alpha}:\mathbb{R} \to \mathbb{R} \) is the interaction potential between the particle pairs, indexed by \(\alpha \in I\), with \( I \coloneqq \{(12),(13),(23)\} \) and \( m_1, m_2, m_3 > 0 \) are the masses of the particles.
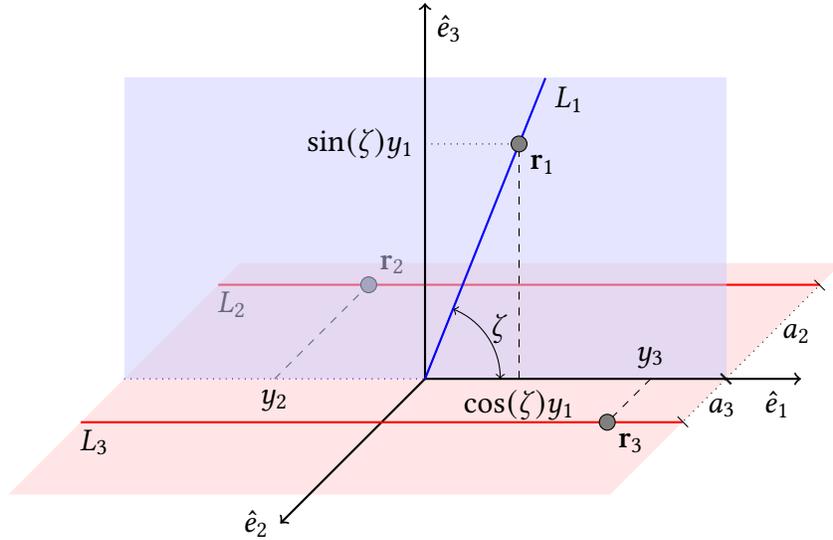
\begin{figure}[!ht]
\centering
\begin{tikzpicture}
    \fill[red!20,opacity=0.5] (-4,0,-4) -- (4,0,-4) -- (4,0,4) -- (-4,0,4) -- cycle;
    \node at (-3.0,-0.1,2.75) [above left] {$L_3$};

    \draw[red, thick,-] (-4,0,-3.25) -- (4,0,-3.25)  ;
    \node at (-3.0,-0.1,-2.0) [above left] {$L_2$};

    \draw [dashed, -] (-2,0,0) -- (-2,0,-3.25);
    \draw[fill=gray] (-2,0,-3.25) circle (3pt) node[anchor=south west] {$\mathbf{r}_2$};

    \fill[blue!20,opacity=0.5] (-4,0,0) -- (4,0,0) -- (4,4,0) -- (-4,4,0) -- cycle;
    \node at (2.25,3.4,0) [above left] {$L_1$};
    \draw [dotted] (0,3.125,0) -- (1.25,3.125,0);

    \draw
    (1,0) coordinate (a) 
    -- (0,0) coordinate (b) 
    -- (1.6,4) coordinate (c) 
    pic["$\zeta$", draw=black, <->, angle eccentricity=1.2, angle radius=1cm]
    {angle=a--b--c};

    \draw[thick,->] (0,0,0) -- (5,0,0) node[anchor=north east] {$\hat{e}_1$};
    \draw[thick,->] (0,0,0) -- (0,5,0) node[anchor=north west] {$\hat{e}_3$};
    \draw[thick,->] (0,0,0) -- (0,0,5) node[anchor=east] {$\hat{e}_2$};

    \draw[blue, dotted] (-4,0,0) -- (0,0,0);
    
    \draw[blue, thick,-] (0,0,0) -- (1.6,4,0);
    \draw [dashed] (1.25,0,0) -- (1.25,3.125,0);
    \draw[red, thick,-] (-4,0,1.5) -- (4,0,1.5)  ;

    \draw [dashed, -] (3,0,0) -- (3,0,1.5);
    \draw[fill=gray] (3,0,1.5) circle (3pt) node[below right] {$\mathbf{r}_3$};

    \draw[fill=gray] (1.25,3.125,0) circle (3pt) node[anchor=north west] 
    {$\mathbf{r}_1$};

    \node at (1.25,0,0) [below] {$\cos(\zeta) y_1$};
    \node at (-2,0,0) [below] {$y_2$};
    \node at (3,0,0) [above] {$y_3$};
    \node at (0,3.125,0) [left] {$\sin(\zeta) y_1 $};

    \draw[dotted, |-|] (4,0,0)--(4,0,1.5);
    \draw[dotted, |-|] (4,0,0)--(4,0,-3.25);
    \node at (4,0,1.0) [right] {$a_3$};
    \node at (4,0,-1.575) [right] {$a_2$};    
\end{tikzpicture}
\caption{Geometrically constrained configuration space of particles.} \label{fig:001}
\end{figure}
We study the system of three bosons given by the operator $H$ in \eqref{3-particle-ham}. 
Regarding the potentials we assume that \( V_{1j} \in L^2_{\text{loc}}(\mathbb{R}^2) \) and \( V_{23} \in L^2_{\text{loc}}(\mathbb{R}) \) and there exist constants \( C, \delta > 0 \) and \( A > 0 \) such that for \( \abs{\mathbf{r}_{\alpha}} > A \).
\begin{equation} \label{short-range}
    \left| V_{\alpha}(\abs{\mathbf{r}_{\alpha}}) \right| \leq C (1 + \abs{\mathbf{r}_{\alpha}})^{-\nu_{\alpha}},
\end{equation}
where \( \nu_{23} \coloneqq 2 + \delta \) and \( \nu_{12} \coloneqq \nu_{13} \coloneqq 3 + \delta \). 
Note that by short--range one typically refers to potentials which decay as \( \left| \cdot \right|^{-2-\delta} \) at infinity. To ensure the applicability of specific decay estimates on the so--called zero--energy resonances of two--particle subsystems, we assume a stronger decay condition on the interaction potentials  
 $V_{12}$ and $V_{13}$. 

In addition, the particles always maintain a minimum distance, making the presence or absence of singularities in the potentials at very small distances irrelevant, except in the special cases where \( a_2 = 0 \) or \( a_3 = 0 \), when the particles $1$ and $2$, respectively, $1$ and $3$, can come arbitrarily close to each other. 
Denote by $\sigma_{\ess}(H)$ the essential and by $\sigma_{\disc}(H)$ the discrete spectrum of $H$. Our main result is
\begin{theorem} \label{main_thrm}
    Let $H$ be the operator defined in equation \eqref{3-particle-ham} with $V_{\alpha}$ fulfilling \eqref{short-range} for any $\alpha\in I=\{(12),(13),(23)\}$. Assume that $\sigma_{\ess}(H) = [0,\infty)$. Then $\sigma_{\disc}(H)$ is at most finite. Contrary to the prediction made in \cite{N:2017}, this system does not exhibit a confinement induced Efimov effect. 
\end{theorem}
\begin{remarks}
    The statement of Theorem \ref{main_thrm} does not impose any 
    conditions on the existence or absence of resonances in 
    two--particle subsystems. 
\end{remarks}
For each $\alpha =(ij) \in I$ we denote by $h_\alpha$ the corresponding two--body Hamiltonian:
\begin{equation} \label{two_particle_operators}
    \begin{split}
            h_{\alpha} &\coloneqq  -\frac{1}{m_i} \frac{\partial^{2}}{\partial y_i^{2}} -  \frac{1}{m_j}\frac{\partial^{2}}{\partial y_j^{2}}  +V_{\alpha}(\abs{\mathbf{r}_{\alpha}}) \, .
    \end{split}
\end{equation}
Let $\Sigma_{ij} \coloneqq \inf \sigma(h_{ij})$ be the bottom of the spectrum of $h_{ij}$ and let $\Sigma \coloneqq \min \{ \Sigma_\alpha : \alpha \in I \}$. Analogously to the HVZ theorem for systems without geometrical constraints \cite[Theorem XIII.17]{RS:BOOK:1978}, we have $\sigma_{\text{ess}}(H) = [\Sigma, \infty)$. Under the conditions of Theorem \ref{main_thrm}, $\Sigma=0$ and consequently $h_\alpha \geq 0$. 

By appropriate rescaling, we can remove the dependence on the masses from the kinetic part of the Hamiltonian $H$. Let $x=(x_1,x_2,x_3)$ with 
\begin{equation}
    x_i= \sqrt{m_i} y_i \text{ for } i \in \{1,2,3\}
\end{equation}
and $\abs{x} \coloneqq (x_1^2+x_2^2+x_3^2)^{1/2}$, then
\begin{equation}\label{3-particle-ham-new_coos}
    H = -\sum_{i=1}^3 \frac{\partial^{2}}{\partial x_i^{2}} + \sum_{\alpha\in I } V_{\alpha}(\abs{\mathbf{r}_{\alpha}}) = -\Delta_x + \sum_{\alpha\in I } V_{\alpha}(\abs{\mathbf{r}_{\alpha}})  
\end{equation}
and
\begin{equation}
    h_{ij} = -\frac{\partial^{2}}{\partial x_i^{2}} - \frac{\partial^{2}}{\partial x_j^{2}} + V_{ij}(\abs{\mathbf{r}_{ij}}) \, .
\end{equation}
In abuse of notation, we denote the transformed operator by the same letter. In this new set of coordinates the distances $\abs{\mathbf{r}_{\alpha}}$ are
\begin{equation} \label{distanc_of_particles}
    \begin{split}
        \abs{\mathbf{r}_{1j}} &= \left( \frac{x_1^2}{m_1} + \frac{x_j^2}{m_j} - 2 \frac{\cos(\zeta)}{\sqrt{m_1m_j}} x_1x_j + a_j^2\right)^{1/2}, \\
        \abs{\mathbf{r}_{23}} &= \left( \left( \frac{x_2}{\sqrt{m_2}} - \frac{x_3}{\sqrt{m_3}} \right)^2 +(a_2-a_3)^2 \right)^{1/2} \, .
    \end{split}
\end{equation}
Note that \(\abs{\mathbf{r}_{1j}}\) remains unchanged under reflection $(x_1,x_j)\mapsto(-x_1,-x_j)$. This symmetry of the potentials will play an important role in our analysis. 
\section{Proof of Theorem \ref{main_thrm}: Absence of the Efimov Effect} \label{proof_of_main_thrm}
   By the min--max principle, it is sufficient to find a finite--dimensional subspace $\mathcal{M} \subset L^2(\R^3)$ such that for any $\psi \in L^2(\R^3)$ orthogonal to $\mathcal{M}$
    \begin{equation}
        \langle \psi, H \psi \rangle \geq 0 \, .
    \end{equation}
    Such a space $\mathcal{M}$ exists, see for example the work by Zhislin \cite{GZ:1974}, if there are constants $b, \tau > 0$ such that
    \begin{equation} \label{L_has_to_be_pos}
    L[\psi] 
      \coloneqq 
        \int_{\R^3} 
          \left( 
            \sum_{i=1}^3 \abs{\partial_{x_i} \psi}^2 
            + \sum_{\alpha \in I} V_{\alpha} \abs{\psi}^2 
          \right) 
        dx 
        - \int_{\abs{x}\in [b,2b]} 
            \frac{\abs{\psi}^2}{\abs{x}^{2+\tau}} 
        dx 
    \geq 0\, 
\end{equation}
for any $\psi \in C^1_0(\R^3)$ with $\supp \psi \subset \{x\in \R^3:\abs{x}>b\}$. We emphasize that no smallness condition of the parameter $\tau>0$ is needed. 
Everywhere below we assume that $b>0$ is sufficiently large and $\tau>0$ is a fixed number which is less than $\delta$, where $\delta>0$ is the parameter in the decay condition on the 
interaction potentials in \eqref{short-range}.
Let
\begin{equation} \label{definition_cones}
    \begin{split}
        K_{1j}(\gamma) 
          &\coloneqq 
            \{ x \in \R^3 : \left(x_1^2+x_j^2\right)^{1/2} \leq \gamma\abs{x} \},
            \quad j\in\{2,3\}\, ,\\
        K_{23}(\gamma) 
          &\coloneqq 
            \left\{ 
              x\in \R^3: \abs{\frac{x_2}{\sqrt{m_2}} 
                 - \frac{x_3}{\sqrt{m_3}}} \leq \gamma \abs{x} 
            \right\}, \\
        \Omega(\gamma) 
          &\coloneqq 
            \R^3 \setminus \left\{  
              K_{12}(\gamma)\cup K_{13}(\gamma) \cup K_{23}(\gamma) 
            \right\}\,.
    \end{split}
\end{equation}
In the following we denote by $\partial K_\alpha(\gamma)$ the boundary of $K_\alpha(\gamma)$. The sets $K_\alpha (\gamma)$ describe parts of the configuration space where the particles $i$ and $j$ 
in $\alpha =(ij)$ are close to each other compared to their distance from the third particle $k\neq\{i,j\}$. 
In Lemma~\ref{cones_are_disjoint} in the Appendix, we show that the sets $K_\alpha(\gamma)$ do not intersect, except for $x=0$, for sufficiently small $\gamma > 0$.

Let $\mathbbm{1}_{A}$ be the indicator function of the set 
$A \subset \R^3$ and define $\psi_\alpha \coloneqq \psi \mathbbm{1}_{K_\alpha(\gamma)}$ and \mbox{$\psi_0 \coloneqq \psi \mathbbm{1}_{\Omega(\gamma)}$}. We prove \eqref{L_has_to_be_pos} 
by estimating for all $\alpha \in I$ the \emph{local energies} 
\begin{equation} \label{L_on_cones}
    \begin{split}
        L_\alpha[\psi_\alpha] 
          &\coloneqq 
            \int_{K_{\alpha}(\gamma)} 
              \left(  
                \abs{\nabla \psi_\alpha}^2 + \sum_{\beta\in I} V_{\beta} \abs{\psi_\alpha}^2 
              \right) 
            dx 
            - \int_{K_{\alpha}(\gamma)} 
                \frac{\abs{\psi_\alpha}^{2}}{\abs{x}^{2+\tau}} 
               dx \, , \\
        L_0[\psi_0] 
          &\coloneqq 
            \int_{\Omega(\gamma)} 
              \left(  
                \abs{\nabla \psi_0}^2 
                + \sum_{\alpha\in I} V_{\alpha} \abs{\psi_0}^2 
              \right) 
            dx 
              - \int_{\Omega(\gamma)} 
                  \frac{\abs{\psi_0}^2}{\abs{x}^{2+\tau}} 
                 dx\, ,
    \end{split}
\end{equation}
and noting 
\begin{equation} \label{expanded_L}
    L[\psi] = L_0[\psi_0] +  \sum_{\alpha \in I } L_\alpha[\psi_\alpha] \, .
\end{equation}
Note that we use a hard cut--off in the definition of the local energies $L_\alpha$ and $L_0$. The analysis of these local energies will involve boundary terms on  $\partial K_{\alpha}(\gamma)$ 
and $\partial \Omega(\gamma)$. The analysis will proceed in several steps. 
As a first step, we show that the functionals $L_\alpha$ for 
$\alpha \in I$ can be bounded in terms of boundary integrals 
over $\partial K_{\alpha}(\gamma)$. 
This is done in the following two lemmas. 
\begin{lemma} \label{L_1j_lemma}
    Fix $\alpha \in \{(12),(13)\}$ and let $P_0[\alpha]$ be the projection in $L^2(\mathbb{R}^3)$ onto functions that are invariant under rotations of the coordinates $(x_1, x_j)$ describing the position of the particle pair $\alpha =(ij)$. Under the conditions of Theorem \ref{main_thrm} there exists a constant $c > 0$, independent of $\psi$, such that
   \begin{equation} \label{1j_surface_est}
       L_{\alpha}[\psi_{\alpha}] \geq -c \int_{\partial K_{\alpha}(\gamma) } \frac{\abs{P_0[\alpha]\psi }^2}{\abs{x}^{1+\tau}} d\sigma \, .
   \end{equation}
\end{lemma}
The statement of Lemma~\ref{L_1j_lemma} is similar to \cite[Lemma 6.7]{BBV:2022}, whereas its proof is significantly more complex. It is based on a detailed analysis of properties of zero--energy resonances of two--dimensional systems which will be done in Section \ref{preliminiaries}. 

For the functional $L_{23}[\psi_{23}]$ we prove the following bound whose proof is similar to the proof of \cite[Theorem 6.1]{BBV:2022}.
\begin{lemma} \label{L_23_lemma}
    Under the conditions of Theorem \ref{main_thrm}, there exists a constant $c > 0$ such that
     \begin{equation} \label{23_surface_est}
          L_{23}[\psi_{23}]\geq -c\int_{\partial K_{23}(\gamma)} \frac{\abs{\psi}^2}{\abs{x}^{1+\tau}} d\sigma \, . \\
     \end{equation}
\end{lemma}

    As the second step, applying the one--dimensional Trace Theorem (see, \cite[Theorem 1, p. 272]{book:Evans}) and Hardy Inequality we show that the right--hand sides of \eqref{1j_surface_est} and \eqref{23_surface_est} can be controlled by a small part of the kinetic energy term on the set $\Omega(\gamma)$. This is done in the following two lemmas.
\begin{lemma} \label{trace_lemma_001}
    For $\gamma_1 \in (\gamma,1)$ and $\alpha \in \{(12),(13)\}$, we define $K_{\alpha}(\gamma, \gamma_1) \coloneqq K_{\alpha}(\gamma_1) \setminus K_{\alpha}(\gamma)$. For any $\varepsilon > 0$, for sufficiently large $b>0$ holds

    \begin{equation} \label{trace_thrm_in_action_001}
        \int_{\partial K_{\alpha}(\gamma) } \frac{\abs{P_0[\alpha]\psi }^2}{\abs{x}^{1+\tau}} d\sigma \leq \varepsilon \int_{K_{\alpha}(\gamma, \gamma_1)} \abs{\nabla \psi }^2 dx \, .
    \end{equation}
\end{lemma}
\begin{lemma} \label{trace_lemma_002}
        Let $\gamma_1 \in (\gamma,1)$ and define $K_{23}(\gamma, \gamma_1) \coloneqq K_{23}(\gamma_1) \setminus K_{23}(\gamma)$. For any $\varepsilon > 0$ we have  
    \begin{equation} \label{trace_thrm_in_action_002}
        \int_{\partial K_{23}(\gamma) } \frac{\abs{\psi}^2}{\abs{x}^{1+\tau}} d\sigma \leq \varepsilon \int_{K_{23}(\gamma, \gamma_1)} \abs{\nabla \psi }^2 dx \, .
    \end{equation}
    for all sufficiently large $b>0$. 
\end{lemma}
Assuming the Lemmas~\ref{L_1j_lemma}, \ref{L_23_lemma}, \ref{trace_lemma_001} and \ref{trace_lemma_002} for the moment, 
we give  the 
\begin{proof}[Proof of Theorem  \ref{main_thrm}:]
Using the bounds of  Lemma~\ref{trace_lemma_001} and \ref{trace_lemma_002} in \eqref{expanded_L} and assuming that $\gamma_1>\gamma$ 
is close enough to $\gamma$ so that the regions $K_{\alpha}(\gamma,\gamma_1)= K_{\alpha}(\gamma_1) \setminus K_{\alpha}(\gamma)$ 
for $\alpha \in I$ do not overlap, we arrive at
\begin{equation} \label{used_the_non_triv_lemmas}
    \begin{split}
        L[\psi]     &\geq  L_0[\psi_0]-\varepsilon \sum_{\alpha \in \{(12),(13)\}}\int_{K_{\alpha}(\gamma, \gamma_1)} \abs{\nabla \psi }^2 dx  - \varepsilon \int_{K_{23}(\gamma, \gamma_1)} \abs{\nabla \psi }^2 dx \\
                    &\geq  L_0[\psi_0]-3\varepsilon  \int_{\Omega(\gamma)} \abs{\nabla \psi }^2 dx \, .
    \end{split}
\end{equation}
where $\varepsilon>0$ is arbitrary small and $b>0$ large. Choosing $b>0$ large enough, Lemma~\ref{potentials_small_outside_cones} shows that each of the potentials satisfies $\abs{V_{\alpha}} \leq C\abs{x}^{-\nu_{\alpha}}$ for some constant $C>0$ on $\Omega(\gamma)\cap\{x: \abs{x}>b\}$. Thus, for fixed $\varepsilon>0$ we have $\abs{V_{\alpha}} \leq \varepsilon \abs{x}^{-2}$ for $x\in \Omega(\gamma)$ with $\abs{x}>b$ and all sufficiently large $b>0$. Hence, 
\begin{equation} \label{up_to_radial_done}
    L_0[\psi_0 ] -3\varepsilon  \int_{\Omega(\gamma)} \abs{\nabla \psi }^2 dx \geq (1-3\varepsilon) \int_{\Omega(\gamma)} \abs{\nabla \psi }^2 dx - \varepsilon \int_{\Omega(\gamma)} \frac{\abs{\psi}^2}{\abs{x}^{2}} dx \, .
\end{equation}
The set $\Omega(\gamma) \subset \R^3$ is conical and applying the radial Hardy Inequality for the last term on the right--hand side of \eqref{up_to_radial_done} (see \ref{radial_hardy_inequality}) yields
\begin{equation}
    L_0[\psi_0 ] -3\varepsilon  \int_{\Omega(\gamma)} \abs{\nabla \psi }^2 dx \geq (1-7\varepsilon) \int_{\Omega(\gamma)} \abs{\nabla \psi }^2 dx  \, .
\end{equation}
This completes the proof of Theorem \ref{main_thrm}.
\end{proof}
\section{Zero--Energy Resonances} \label{preliminiaries}
To prove Lemma~\ref{L_1j_lemma} we will need some properties of zero--energy resonances of two--particle Schr\"odinger Operators in dimension two. The most important of them are the estimates on the decay of such resonances, which will be given in Lemma~\ref{first_lemma}.

Let
\begin{equation} \label{schrodinger}
    h=-\Delta+V \text{ on } L^2(\R^2)
\end{equation}
with $V$ satisfying \eqref{short-range} with parameter $\nu_{\alpha}=3+\delta$ for some $\delta>0$ and $V(x)=V(-x)$. Following \cite{Y:1975}  we say that $h$ has a virtual level (zero--energy resonance) if $h\geq 0$ and for any $\varepsilon>0$, $h+\varepsilon \Delta$ has an eigenvalue below zero.   

Let $\dot H^1(\R^2)$ be the homogeneous Sobolev space defined as 
\begin{equation}
    \begin{split}
        \dot H^1(\R^2) &\coloneqq \{u \in L^2_{\loc}(\R^2): \nabla u \in L^2(\R^2)  \},
    \end{split}
\end{equation}
equipped with the norm
\begin{equation}
        \norm{u}_{\dot H^1(\R^2)} \coloneqq \left( \int_{\R^2}\abs{\nabla u}^2 dx +  \int_{\abs{x}\leq 1}\abs{ u}^2 dx \right)^{1/2} \, .
\end{equation}
 We will use the following result of \cite[Theorem 2.2]{BBV:2022}
\begin{lemma} \label{virtual_level_lemma}
    Assume that $h$ has a virtual level. Then 
    \begin{enumerate}
    \item there exists a unique non--negative $\varphi_0 \in \dot H^1(\R^2)$ with $\norm{\varphi_0}_{\dot H^1(\R^2)} = 1$ such that for any $\psi \in  \dot H^1(\R^2) $
        \begin{equation} \label{generalized_zero_solution}
            \langle \nabla \psi, \nabla \varphi_0 \rangle + \langle \psi, V \varphi_0 \rangle = 0 \, . 
        \end{equation}
    \item there exists $\mu>0$ such that for any $\psi \in H^1(\R^2)$ with $\langle \nabla \psi, \nabla \varphi_0 \rangle =0 $ 
        \begin{equation} \label{ortho_to_zero}
            \langle \psi, h \psi \rangle \geq \mu \norm{\nabla \psi}^2 \, . 
        \end{equation}
\end{enumerate}
\end{lemma}
\begin{remark}
    Note that, in general, $\varphi_0 \notin L^2(\R^2)$. Moreover, if the potential $V$ is radially symmetric and compactly supported, it is easy to see that $\varphi_0$ is also a radially symmetric function that does not decay at infinity. If $V$ is not radially symmetric, then $\varphi_0$ is not radially symmetric either. For this case, we prove, that if $V$ is symmetric under reflection, then the projection of $\varphi_0$ onto the subspace orthogonal to radially symmetric functions is in $L^2(\R^2)$. The proof is given in the next lemma.
\end{remark}
Let $P_0$ the projection onto radially symmetric functions in $L^2(\R^2)$ and  \mbox{$P_\perp \coloneqq 1-P_0$} the projection onto its orthogonal complement. The next result shows that even though \mbox{$\varphi_0\notin L^2(\R^2)$}  its projection $P_\perp \varphi_0$ is in a weighted $L^2$--space. 
\begin{lemma} \label{first_lemma}
    Let $\varphi_0$ be a virtual level of $h$, then there exists $l(\delta)>0$  such that
        \begin{equation} \label{Perp_in_L2}
            (1+\abs{\cdot})^{l}P_\perp \varphi_0 \in L^2(\R^2) \, .
        \end{equation}
\end{lemma}

\subsection{Proof of Lemma \ref{first_lemma}}
Let $f\coloneqq P_0 \varphi_0$ and $g \coloneqq P_\perp \varphi_0 $. 
Since \( V(x) = V(-x) \), a virtual level \(\varphi_0\) can be either an even or odd function with respect to reflection. However, by Lemma~\ref{virtual_level_lemma}, \(\varphi_0\) is non--negative, which implies that it must be an even function. Consequently, for almost all $\abs{x}$, 
\begin{equation}
    \int_{0}^{2\pi} e^{\pm i \theta} g(\abs{x},\theta) d\theta =0 \, .
\end{equation}
Note that for all functions $F(\abs{x},\theta) \in \dot H^1(\R^2)$ orthogonal to radially symmetric functions with
\begin{equation} \label{higher_momentum}
    \int_{0}^{2\pi} e^{\pm i \theta} F(\abs{x},\theta) d\theta =0 
\end{equation}
the following inequality holds:
\begin{equation} \label{Hardy_type_ineq}
   \int_{\R^2} \abs{\nabla F}^2 \, dx \geq 4 \int_{\R^2} \frac{F(x)^2}{\abs{x}^2} \, dx \, .
\end{equation}
In particular \eqref{Hardy_type_ineq} holds for $F=g$. To prove \eqref{Perp_in_L2} it suffices now to show that 
\begin{equation} \label{grad_est_g}
    \nabla \left( (1+\abs{\cdot})^{l +1 }  g \right) \in L^2(\R^2)\,  .
\end{equation}
Choose $\xi \in C^\infty([0,\infty))$ with $\xi(t) = 0$ for $t\leq 1$ and $\xi(t)=1$ for $t \geq 2$, such that $\xi(t)\leq 1$ and $\xi'(t) \leq 2$ for any $t\in [0,\infty)$. For any $\omega,\kappa,\beta >0$ we define
\begin{equation} \label{def_of_G}
     G(\abs{x})  \coloneqq \frac{\abs{x}^\kappa}{1+\omega \abs{x}^\kappa} \xi(\abs{x}/\beta)\, .
\end{equation}
Inserting $\psi = G^2 g$ into  \eqref{generalized_zero_solution}  and writing $\varphi_0 = f+g$ yields
\begin{equation} \label{starting_point}
    0= \langle G^2 g, h f \rangle + \langle G^2 g, h g \rangle \, .
\end{equation}
Since $P_0$ and $P_\perp$ commute with $-\Delta$, it holds
\begin{equation} \label{used_orto}
        \langle G^2 g, h f\rangle =\langle G g, V G f\rangle  \, .
\end{equation}
Observe that 
\begin{equation} \label{IMS}
    \begin{split}
        \langle G^2  g, h  g\rangle &=  \langle G  g, h G g\rangle  - \left\langle  G g, \frac{\abs{\nabla G}^2}{G^2}  G g\right\rangle \\
        &= \norm{\nabla(Gg)}^2 + \langle G  g, V G g\rangle  - \left\langle  G g, \frac{\abs{\nabla G}^2}{G^2}  G g\right\rangle\, .
    \end{split}
\end{equation}
Combining \eqref{starting_point}, \eqref{used_orto} and \eqref{IMS} we arrive at
\begin{equation} \label{go_on_from_here}
     \norm{\nabla(Gg)}^2 + \langle G  g, V G g\rangle  - \left\langle  G g, \frac{\abs{\nabla G}^2}{G^2}  G g\right\rangle + \langle G g, V G f\rangle = 0\, .
\end{equation}
We claim, there exists some $\varepsilon>0$ and a constant $c(\beta)>0$ that both are independent of $\omega$ such that 
\begin{equation} \label{wish_for}
    \langle G  g, V G g\rangle  - \left\langle  G g, \frac{\abs{\nabla G}^2}{G^2}  G g\right\rangle + \langle G g, V G f\rangle \geq - (1-\varepsilon) \norm{\nabla(Gg)}^2  -c(\beta) \norm{\varphi_0}_{\dot H^1}^2 \, .
\end{equation}
Assuming this claim for the moment, we complete the proof of Lemma~\ref{first_lemma}. Combining \eqref{go_on_from_here} with \eqref{wish_for} yields
\begin{equation}
    \varepsilon \norm{\nabla(Gg)}^2 \leq c(\beta) \norm{\varphi_0}_{\dot H^1}^2
\end{equation}
and consequently $\norm{\nabla(Gg)}^2$ is bounded uniformly in $\omega$, which proves \eqref{grad_est_g} with $\kappa=l+1$. Then taking the limit $\omega \to 0$ concludes the proof of Lemma~\ref{first_lemma}. \qed

\smallskip

We prove the remaining statement in \eqref{wish_for} by estimating each of the terms on the left--hand side of \eqref{wish_for} separately.
\subsection{First Term in \eqref{wish_for}}
The function $G$ vanishes for $\abs{x}<\beta$ and $V$ fulfills \eqref{short-range} and therefore there exists $C>0$ such that
\begin{equation} \label{first_term_with_pot}
   \abs{ \langle G  g, V G g\rangle } \leq \frac{C}{(1+\beta)^{1+\delta}} \int_{\abs{x}>\beta} \frac{\abs{Gg}^2}{\abs{x}^2} dx \leq \frac{4C}{(1+\beta)^{1+\delta}} \norm{ \nabla (Gg)}^2 \, .
\end{equation}
\subsection{Second Term in \eqref{wish_for}}
In Lemma~\ref{Lemma_G_derivative} in the appendix we show for $\abs{x}>2\beta$ that
\begin{equation} \label{est_on_G}
        \abs{\nabla G(x)} \leq  \kappa \abs{x}^{-1} G(x) \quad \text{for} \quad \abs{x}>2\beta
\end{equation}
and
\begin{equation} \label{uniform_bound}
    \abs{\nabla G(x)}^2 \leq \beta^{\kappa -1 } \left( 2^{\kappa+1} +  \kappa 2^{\kappa-1} \right) \eqqcolon c_1(\beta) \quad \text{for} \quad \abs{x} \in [\beta, 2\beta] \, .
\end{equation}
Recall that the function $Gg$ is orthogonal to radially symmetric functions and in addition satisfies \eqref{higher_momentum}, consequently \eqref{est_on_G} together with \eqref{Hardy_type_ineq} yields
\begin{equation} \label{conclude_from_here}
    \int_{\abs{x}\geq 2\beta} \frac{\abs{\nabla G}^2}{G^2} \abs{Gg}^2 dx   \leq \kappa^2  \int_{\abs{x}\geq 2\beta} \abs{x}^{-2} \abs{Gg}^2 dx  \leq    \frac{\kappa^2}{4} \norm{\nabla(Gg)}^2  \, .
\end{equation}
By combining \eqref{uniform_bound} and \eqref{conclude_from_here}, we obtain 
\begin{equation} \label{splitted_nabla_G_term}
    \begin{split}
        \left\langle  G g, \frac{\abs{\nabla G}^2}{G^2}  G g\right\rangle &\leq  \int_{\beta \leq \abs{x}\leq 2\beta} \abs{\nabla G}^2 \abs{g}^2 dx + \int_{ \abs{x}\geq  2\beta } \abs{\nabla G}^2 \abs{g}^2 dx \\
        &\leq c_1(\beta)  \int_{\beta \leq \abs{x}\leq 2\beta} \abs{g}^2 dx +  \frac{\kappa^2}{4} \norm{\nabla(Gg)}^2 \, .
    \end{split}
\end{equation}
Applying \eqref{Hardy_type_ineq} gives
\begin{equation}\label{g_on_compact}
    \begin{split}
        \int_{\beta \leq \abs{x}\leq 2\beta } \abs{g}^2 dx &\leq  (2\beta)^2 \int_{\beta\leq \abs{x}  \leq 2\beta} \abs{x}^{-2}\abs{g}^2 dx \\
        &\leq 4(2\beta)^2 \int_{\R^2} \abs{\nabla g}^2 dx \\
        &\leq  4(2\beta)^2 \norm{\varphi_0}_{\dot H^1}^2  \, .
    \end{split}
\end{equation}
Inserting \eqref{g_on_compact} into \eqref{splitted_nabla_G_term} yields
\begin{equation} \label{estimate_on_nabla_G_term}
    \left\langle  G g, \frac{\abs{\nabla G}^2}{G^2}  G g\right\rangle \leq    \frac{\kappa^2}{4} \norm{\nabla(Gg)}^2  + c_2(\beta) \norm{\varphi_0}_{\dot H^1}^2 
\end{equation}
where $c_2(\beta) = 4(2\beta)^2 c_1(\beta)$.
\subsection{Third Term in \eqref{wish_for}}
Let $\tilde \varepsilon>0$. By \eqref{short-range} and Schwarz Inequality, we obtain
\begin{equation}\label{split_inequal}
    \begin{split}
        \abs{\langle G  g, V G f \rangle} &\leq C \int_{\abs{x}>\beta} \frac{\abs{  G  g}}{\abs{x}^{1+\delta/2}}  \frac{\abs{G f }}{\abs{x}^{2+\delta/2}}dx \\
        &\leq C \tilde \varepsilon  \int_{\abs{x}>\beta} \frac{\abs{G}^2}{\abs{x}^{2+\delta}} \abs{ g }^2 dx  + \frac{C}{\tilde \varepsilon} \int_{\abs{x}>\beta} \frac{\abs{G}^2}{\abs{x}^{4+\delta}} \abs{ f }^2 dx \, .
    \end{split}
\end{equation}
Using \eqref{Hardy_type_ineq} for $Gg$ in the first term on the right--hand side of \eqref{split_inequal} we get
\begin{equation} \label{hardy_again_with_eps2}
     C \tilde \varepsilon  \int_{\abs{x}>\beta} \frac{\abs{G}^2}{\abs{x}^{2+\delta}} \abs{ g }^2 dx \leq\frac{C \tilde \varepsilon}{4} \norm{\nabla(Gg)}^2 \, .
\end{equation}
To estimate the second term on the right--hand side of \eqref{split_inequal} we choose $\kappa$ in the definition of $G$ in \eqref{def_of_G} as $ 1 + \delta/4$, then for every $\abs{x}>\beta$
\begin{equation} \label{bound_on_G}
    \frac{\abs{G(x)}^2}{\abs{x}^{4+\delta}}  \leq \abs{x}^{2\kappa -4-\delta} \abs{\xi(\abs{x}/\beta)}^2\leq \abs{x}^{-2-\delta/2} \abs{\xi(\abs{x}/\beta)}^2 \, .
\end{equation}
Applying \eqref{bound_on_G} yields
\begin{equation} \label{used_bound_on_G}
    \begin{split}
    \int_{\abs{x}>\beta} \frac{\abs{G}^2}{\abs{x}^{4+\delta}} \abs{ f }^2 dx &\leq \int_{\R^2}  \abs{x}^{-2-\delta/2}  \abs{ \xi(\abs{x}/\beta) f(x) }^2 dx\, .
    \end{split} 
\end{equation}
Note that the function $\xi(\abs{x}/\beta) f(x) $ vanishes for $\abs{x}=\beta$  and consequently for $\beta>0$ large enough we can apply the two--dimensional Hardy Inequality (see Lemma~\ref{2d_2_hardy_inq}) and get
\begin{equation} \label{two_dim_hardy}
    \int_{\R^2}  \abs{x}^{-2-\delta/2}  \abs{ \xi(\abs{x}/\beta) f(x) }^2 dx  \leq \int_{\R^2} \abs{\nabla (\xi(\abs{x}/\beta) f(x) )}^2 dx \, .
\end{equation}
Since $\nabla \xi(\abs{\cdot}/\beta)$ is supported in $\beta\leq \abs{x} \leq 2\beta$ and $\abs{\nabla \xi(\abs{\cdot}/\beta)} \leq 2/\beta$ we get for the right--hand side of \eqref{two_dim_hardy}
\begin{equation} \label{only_f_terms}
    \int_{\R^2} \abs{\nabla (\xi(\abs{x}/\beta) f(x) )}^2 dx \leq \frac{8}{\beta^2} \int_{\beta<\abs{x}<2\beta} \abs{f}^2 dx + 2 \int_{\abs{x}>\beta} \abs{\nabla f}^2 dx 
\end{equation}
where we have used that $(a+b)^2<a^2+b^2$. The function $\varphi_0 \in L^2_{\loc}(\R^2)$ is normalized with respect to $\dot H^1(\R^2)$. Then, due to the orthogonality of $f$ and $g$ for fixed $\abs{x}$, there exists a constant $c_3(2\beta)>0$ such that
\begin{equation} \label{norms_varphi_0}
    \norm{f}_{\abs{x}<2\beta}^2 \leq  \norm{\varphi_0}_{\abs{x}<2\beta}^2 \leq c_3(2\beta)  \norm{\varphi_0}^2_{\dot H^1(\R^2)}\, .
\end{equation}
Combining \eqref{only_f_terms} and \eqref{norms_varphi_0} shows there exits a constant $c_4(\beta)>0 $, such that
\begin{equation} \label{only_norm_of_varphi_0}
    \int_{\R^2} \abs{\nabla (\xi(\abs{x}/\beta) f(x) )}^2 dx \leq c_4(\beta) \norm{\varphi_0}^2_{\dot H^1(\R^2)} \, .
\end{equation}
Relations \eqref{used_bound_on_G}, \eqref{two_dim_hardy} and \eqref{only_norm_of_varphi_0} imply 
\begin{equation}  \label{H_1_dot_varphi_0}
    \int_{\abs{x}>\beta} \frac{\abs{G}^2}{\abs{x}^{4+\delta}} \abs{ f }^2 dx \leq c_4(\beta)  \norm{\varphi_0}^2_{\dot H^1} \, .
\end{equation}
Substituting \eqref{hardy_again_with_eps2} and \eqref{H_1_dot_varphi_0}  into \eqref{split_inequal} gives
\begin{equation} \label{third_term}
    \abs{\langle G  g, V G f \rangle} \leq \frac{C \tilde \varepsilon}{4} \norm{ \nabla (Gg)}^2 + c_5(\beta, \tilde \varepsilon) \norm{\varphi_0}_{\dot H^1}^2 ,
\end{equation}
where 
\begin{equation}
    c_5(\beta,\tilde \varepsilon) \coloneqq \frac{C}{\tilde \varepsilon} \cdot
 c_4(\beta)
\end{equation}
is a constant depending on $\beta$ and $\tilde \varepsilon$. Note that $C>0$ is a constant depending on the potential $V$ only and the parameter $\tilde \varepsilon>0$ can be chosen small.

\subsection{Completing the Proof of \eqref{wish_for}}
Combining the inequalities \eqref{first_term_with_pot}, \eqref{estimate_on_nabla_G_term} and \eqref{third_term} yields for a constant $c(\beta,\tilde \varepsilon)>0$ that depends on $\beta$ and $ \tilde \varepsilon$ but is independent of $\psi$ 
\begin{equation}
    \begin{split}
        \langle G  g, V G g\rangle  &- \left\langle  G g, \frac{\abs{\nabla G}^2}{G^2}  G g\right\rangle + \langle G g, V G f\rangle\\
        &\geq -\left( \frac{\kappa^2}{4}  + \frac{4C}{(1+\beta)^{1+\delta}} + \frac{C\tilde \varepsilon}{4}\right) \norm{\nabla(Gg)}^2  -c(\beta,\tilde \varepsilon) \norm{\varphi_0}_{\dot H_1}^2 \, .
    \end{split}
\end{equation}
Since we can always assume $\delta<1$ and since $\kappa=1+\delta/4$ we have $\kappa^2/4 <1$. Consequently for $\beta>0$ sufficiently large and by assuming that $\tilde \varepsilon>0$ in \eqref{split_inequal} is chosen to be small we can have
\begin{equation}
    \left( \frac{\kappa^2}{4}  + \frac{4C}{(1+\beta)^{1+\delta}} + \frac{C\tilde \varepsilon}{4}\right) < 1-\varepsilon \, .
\end{equation}
The constant $c(\beta, \tilde \varepsilon)$ for fixed $\beta$ and  $\tilde \varepsilon$ may be large but is finite and independent on $\omega$. As explained earlier taking the limit $\omega \to 0$ completes the proof of Lemma~\ref{first_lemma}. 
\section{Proofs of the Lemmas for the Main Theorem} \label{proofs}
In this section, we prove the lemmas stated in  Section \ref{proof_of_main_thrm}.
\subsection{Proof of Lemma~\ref{L_1j_lemma} } \label{proof_of_L_1j_lemma}
We show the statement for $L_\alpha[\psi_\alpha]$ with $\alpha=(12)$. The proof for $\alpha=(13)$ is similar. We drop the index $\alpha$ whenever possible.
\begin{remark} The proof of Lemma~\ref{L_1j_lemma} is organized as follows. First, we introduce several new functions and state three lemmas that correspond to the main steps in the proof, showing how they conclude Lemma~\ref{L_1j_lemma}. The proofs of these three lemmas are then provided in Section \ref{three_lemmatas}. \end{remark}
Due to \cite[Lemma 5.1]{VZ:1982} for given $\varepsilon>0$ and fixed $\gamma>0$ there exists a $\tilde{\gamma} \in (0,\gamma)$ and a piecewise continuously differentiable function $u:\R^3 \to [0,1]$ with
\begin{equation} \label{u_and_v_functions}
    u(x) = \begin{cases}
        1 &x\in K_{12}(\tilde \gamma) \\
        0 &x \notin K_{12}(\gamma)
    \end{cases}
\end{equation}
such that for $v\coloneqq \left(1-u^2\right)^{1/2}$
\begin{equation} \label{localization_error}
    \abs{\nabla u}^2 + \abs{\nabla v}^2 \leq \varepsilon \left( \frac{v^2}{\abs{x}^2} + \frac{u^2}{\abs{(x_1,x_2)}^2} \right) 
\end{equation}
for every $x=(x_1,x_2,x_3)\in \R^3$.

Let
\begin{equation} \label{def_tilde_psi}
    \begin{split}
        \psi_1 &\coloneqq (P_\perp \psi_{12})  v,\\
        \psi_2 &\coloneqq (P_\perp \psi_{12}) u + P_0 \psi_{12} \, .
    \end{split}
\end{equation}
Note that we use smooth localization for the function $P_\perp \psi_{12}$ which allows us to apply the IMS--Localization formula. For reasons that will be explained later, we can not do such a smooth localization for the function $P_0 \psi_{12}$.

As the first step, we show that $L_{12}[\psi_{12}]$ can be estimated in terms of integrals involving $\psi_2$ only. Namely, we prove the following:
\begin{lemma} \label{first_part_of_lem21}
    Let $K_{12}(\gamma,\tilde \gamma) \coloneqq K_{12}(\gamma) \setminus K_{12}(\tilde \gamma)$ and $b>0$ large enough. Then
    \begin{equation}
        \begin{split}
            L_{12}[\psi_{12}] =  \int_{K_{12}(\gamma)} \left( \abs{\nabla \psi_{12}}^2 + \sum_{\alpha\in I} V_{\alpha} \abs{\psi_{12}}^2 \right) dx -  \int_{K_{12}(\gamma)\setminus S(0,b)} \frac{\abs{\psi_{12}}^{2}}{\abs{x}^{2+\tau}} dx  \geq \tilde{L}_{12}[\psi_2]
        \end{split}
    \end{equation}
    where 
    \begin{equation} \label{L_tilde_def}
    \begin{split}
        \tilde{L}_{12}[\psi_2] \coloneqq &\int_{K_{12}(\gamma)} \left( \abs{\nabla \psi_2}^2 + \sum_{\alpha\in I} V_{\alpha} \abs{\psi_2}^2 \mathbbm{1}_{K_\alpha(\tilde \gamma)}  -2\frac{\abs{\psi_2}^{2}}{\abs{x}^{2+\tau}} \right)dx \\
        &-\varepsilon\int_{K_{12}(\gamma,\tilde \gamma)} \frac{\abs{P_\perp {\psi}_2}^2}{\abs{(x_1,x_2)}^2} dx \, .
    \end{split}
\end{equation}
\end{lemma}
As the next step we extend $\psi_2$ for fixed $x_3$ outside of $K_{12}(\gamma)$. Note that $\psi_2 = P_0 \psi_{12}$ on $\partial K_{12}(\gamma)$ and thus is constant for fixed $x_3 \in \R$ on $\partial K_{12}(\gamma)$. This allows us to continuously extend $\psi_2$  to $\R^3$ by a function which does not depend on $(x_1,x_2)$ outside of $K_{12}(\gamma)$. Let $\tilde \psi_2$ be this new function, then since  $\psi \in C^1_0(\R^3)$ it follows that
\begin{equation} \label{space_for_tilde_psi_2}
    \tilde \psi_2(\cdot, x_3) \in \dot{H}^1(\R^2), \quad \tilde \psi_2(x_1,x_2,\cdot) \in H^1(\R) \quad \text{and} \quad \tilde \psi_2 \in \dot H^1(\R^3) \, .
\end{equation}

Denote by $\nabla_{12} =(\partial_{x_1},\partial_{x_2})$ the gradient in the $(x_1,x_2)$--plane. If $h_{12}$ has a virtual level let $\varphi_0 \in \dot H^1(\R^2)$ be the corresponding solution of $h_{12} \varphi_0 =0$. We normalize $\varphi_0$ with respect to the seminorm corresponding to the sesquilinearform 
\begin{equation}
    \langle \nabla_{12} f, \nabla_{12} \, g \rangle_\ast = \int_{\R^2} \left(  \nabla_{12} f \cdot \nabla_{12} g \right) \, d(x_1,x_2), \quad f,g \in H^1(\R^2)
\end{equation}
such that
\begin{equation}
    \norm{\varphi_0}_\ast^2 = \int_{\R^2} \abs{ \nabla \varphi_0}^2 d(x_1,x_2)= 1 \, .
\end{equation}
For every $x_3\in \R$ let
\begin{equation} \label{def_prjector_weight}
    \Phi(x_3) \coloneqq \langle \nabla_{12} \varphi_0, \nabla_{12} \tilde \psi_{2} \rangle_{L^2(\R^2)} \,  .
\end{equation}
We show that the function $\Phi$ is in $L^2(\R)$. Since $\nabla_{12}\tilde \psi_2$ vanishes outside of $\supp{\psi_{12}}$ and consequently outside of $\supp{\psi}$, applying Schwarz Inequality yields for some constant $C>0$
\begin{equation} \label{Phi_is_l2}
    \begin{split}
        \norm{\Phi}_{L^2(\R)}^2 &= \int \abs{ \iint \left( \nabla_{12} \varphi_0 \mathbbm{1}_{\supp( \psi)} \cdot \nabla_{12} \tilde \psi_2 \right)d(x_1,x_2)}^2 dx_3\\
        &\leq \int_{\R} \norm{\nabla_{12} \varphi_0 \mathbbm{1}_{\supp( \psi)}}_{L^2(\R^2)} \cdot  \norm{\nabla_{12} \tilde \psi_2(\cdot, x_3)}_{L^2(\R^2)} dx_3 \\
        &\leq \frac{1}{2} \int_{\R} \left(  \norm{\nabla_{12} \varphi_0 \mathbbm{1}_{\supp( \psi)}}_{L^2(\R^2)}^2 + \norm{\nabla_{12} \tilde \psi_2(\cdot, x_3)}_{L^2(\R^2)}^2 \right) dx_3 \\
        &\leq C \norm{ \nabla_{12} \varphi_0}_{\ast} +  \frac{1}{2} \norm{ \tilde \psi_2}_{\dot H^1(\R^3)} \, .
    \end{split}
\end{equation}
In the last line of \ref{Phi_is_l2} we have used, that $\psi$ is compactly supported. 

Let $F(x_1,x_2,x_3)$ be defined by
\begin{equation} \label{projection}
\tilde \psi_{2}(x_1,x_2,x_3) =   \varphi_0(x_1,x_2) \Phi(x_3) + F(x_1,x_2,x_3) \, .
\end{equation}
For almost all $x_3\in \R$ the function $F$ satisfies 
\begin{equation} \label{orthogonality_in_grad_sense}
    \langle \nabla_{12} \varphi_0, \nabla_{12} F\rangle_{L^2(\R^2)} =0 \, .
\end{equation} 
If the virtual level $\varphi_0$ does not exist we assume $\Phi(x_3) \equiv 0$ and consequently $F= \tilde \psi_2$.
\begin{remark}
    For fixed $x_3 \in \R$, the expression \eqref{projection} is a projection of $\tilde \psi_2$ onto the virtual level $\varphi_0$ within $\dot H^1(\R^2)$. Note that, unlike $\tilde \psi_2(\cdot, x_3)$, the function $\psi_2(\cdot, x_3)$ is not in $\dot H^1(\R^2)$ and therefore such a projection would not be possible. That is the reason why we needed to extend $\psi_2$ introducing $\tilde \psi_2$. 
\end{remark}
With these definitions in  \eqref{def_prjector_weight} and \eqref{projection}, we can state the following:
\begin{lemma}\label{second_part_of_lem21}
    For the functional $\tilde{L}_{12}$ we have the estimate
    \begin{equation} \label{eq_in_second_part_of_lem21}
    \begin{split}
        \tilde{L}_{12}[\psi_2] \geq &\mu \norm{ \nabla_{12} F}^2  + \frac{1}{2}\int_{K_{12}(\gamma)} \abs{\partial_{x_3} (P_\perp \psi_2)}^2 dx -\varepsilon\int_{K_{12}(\gamma,\tilde \gamma)} \frac{\abs{P_\perp {\psi}_2}^2}{\abs{(x_1,x_2)}^2} dx \\
        &- C\int_{\partial K_{12}(\gamma) } \frac{\abs{P_0 \psi }^2}{\abs{x}^{1+\tau}} d\sigma\, 
    \end{split}
\end{equation}
where $\mu>0$ is the parameter in assertion (2) of Lemma~\ref{virtual_level_lemma} and $C>0$ a constant independent of $\psi$.
\end{lemma}
\begin{remark}
    By comparing the statement of Lemma~\ref{L_1j_lemma} with Lemma~\ref{second_part_of_lem21} we see that the assertion of Lemma~\ref{L_1j_lemma} follows immediately if we can show that the sum of the first three terms on the right--hand side of \eqref{eq_in_second_part_of_lem21} are positive.  We give this in the following lemma which is the last step in the proof of Lemma~\ref{L_1j_lemma}.
\end{remark}
\begin{lemma}\label{third_part_of_lem21} For $\varepsilon \in (0,\mu/8)$ and $b>0$ large enough we have
    \begin{equation} \label{last_of_42}
    \mu \norm{ \nabla_{12} F}^2  + \frac{1}{2}\int_{K_{12}(\gamma)} \abs{\partial_{x_3} (P_\perp \psi_2)}^2 dx -\varepsilon\int_{K_{12}(\gamma,\tilde \gamma)} \frac{\abs{P_\perp {\psi}_2}^2}{\abs{(x_1,x_2)}^2} dx \geq 0 \, .
\end{equation}
\end{lemma}
\begin{remark}
    Recall that for any fixed $\gamma>0$ we can make $\varepsilon>0$ arbitrarily small by choosing $\tilde \gamma \in (0, \gamma)$ small. In particular, we can always assume $\varepsilon < \mu/8$.
\end{remark}
\subsection{Proofs of Lemmas~\ref{first_part_of_lem21}, \ref{second_part_of_lem21} and \ref{third_part_of_lem21}} \label{three_lemmatas}
\subsubsection{Proof of Lemma~\ref{first_part_of_lem21}} \label{prof_of_lem_51} We aim to decompose the expression
\begin{equation} \label{l_12_copy}
    L_{12}[\psi_{12}] = \int_{K_{12}(\gamma)} \left( \abs{\nabla \psi_{12}}^2 + \sum_{\alpha\in I} V_{\alpha} \abs{\psi_{12}}^2 \right) dx -  \int_{K_{12}(\gamma)\setminus S(0,b)} \frac{\abs{\psi_{12}}^{2}}{\abs{x}^{2+\tau}} dx
\end{equation}
into terms that involve either $\psi_1$ or $\psi_2$ defined in \eqref{def_tilde_psi}. In the region $K_{12}(\gamma, \tilde{\gamma})$, the potentials $V_\alpha$ satisfy \mbox{$\abs{V_\alpha} \leq \tilde{C} \abs{x}^{-2-\delta}$} for some constant $\tilde{C} > 0$, assuming $\gamma > 0$ is sufficiently small and $b > 0$ is sufficiently large (see Lemma~\ref{potentials_small_outside_cones}). Since $\tau<\delta$ we have
\begin{equation}
    \int_{K_{12}(\gamma, \tilde{\gamma})} \sum_{\alpha\in I} V_\alpha \abs{\psi_{12}}^2 dx \geq - \int_{K_{12}(\gamma)\setminus S(0,b)} \frac{\abs{\psi_{12}}^{2}}{\abs{x}^{2+\tau}} dx
\end{equation}
and consequently
\begin{equation} \label{l_12_copy_copy}
    L_{12}[\psi_{12}] \geq \int_{K_{12}(\gamma)} \left( \abs{\nabla \psi_{12}}^2 + \sum_{\alpha\in I} V_{\alpha} \abs{\psi_{12}}^2 \mathbbm{1}_{K_{12}(\tilde \gamma)}  \right) dx -  2\int_{K_{12}(\gamma)\setminus S(0,b)} \frac{\abs{\psi_{12}}^{2}}{\abs{x}^{2+\tau}} dx \, .
\end{equation}
Due to the orthogonality of $P_0 \psi_{12}$ and $P_\perp \psi_{12}$
\begin{equation} \label{ortho_decomp_kinetic}
    \int_{K_{12}(\gamma)} \abs{\nabla \psi_{12}}^2 dx = \int_{K_{12}(\gamma)} \abs{\nabla (P_\perp \psi_{12}) }^2 dx + \int_{K_{12}(\gamma)} \abs{\nabla  (P_0 \psi_{12})}^2 dx \, .
\end{equation}
With the bounds on the localization estimates in \eqref{localization_error} we have
\begin{equation} \label{kin_perp_with_loc_error}
    \begin{split}
        \int_{K_{12}(\gamma)} \abs{\nabla (P_\perp \psi_{12}) }^2 dx   \geq &\int_{K_{12}(\gamma)} \abs{\nabla( P_\perp \psi_{12} u) }^2 dx  - \varepsilon \int_{K_{12}(\gamma,\tilde \gamma)} \abs{P_\perp \psi_{12}}^2 \frac{u^2}{\abs{(x_1,x_2)}^2} dx \\
        &+ \int_{K_{12}(\gamma)} \abs{\nabla (P_\perp \psi_{12}v) }^2 dx - \varepsilon \int_{K_{12}(\gamma,\tilde \gamma)} \abs{P_\perp \psi_{12}}^2  \frac{v^2}{\abs{x}^2}  dx,
    \end{split}
\end{equation}
where $u,v$ are functions defined by \eqref{u_and_v_functions}.  Due to the definition of functions $\psi_1$ and $\psi_2$ in \eqref{def_tilde_psi} we have
\begin{equation} \label{der_psi_2}
    \int_{K_{12}(\gamma)} \abs{\nabla( P_\perp \psi_{12} u) }^2 dx +\int_{K_{12}(\gamma)} \abs{\nabla  (P_0 \psi_{12})}^2 dx=\int_{K_{12}(\gamma)}  \abs{\nabla  \psi_2}^2  dx
\end{equation}
and since $\psi_1$ vanishes on $K_{12}(\tilde \gamma)$ it follows
\begin{equation} \label{der_psi_1}
    \int_{K_{12}(\gamma)} \abs{\nabla (P_\perp \psi_{12}v) }^2 dx =\int_{K_{12}(\gamma,\tilde \gamma)}  \abs{\nabla  \psi_1 }^2  dx \, .
\end{equation}
Inserting \eqref{kin_perp_with_loc_error}  into  \eqref{ortho_decomp_kinetic} and applying the relations in \eqref{der_psi_2} and \eqref{der_psi_1} we arrive at
\begin{equation} \label{IMS_and_ortho_first}
    \begin{split}
        \int_{K_{12}(\gamma)}  \abs{\nabla \psi_{12}}^2  dx \geq &\int_{K_{12}(\gamma,\tilde \gamma)}  \abs{\nabla  \psi_1 }^2  dx - \varepsilon\int_{K_{12}(\gamma,\tilde \gamma)} \frac{\abs{\psi_1}^2}{\abs{x}^2} dx \\
        &+ \int_{K_{12}(\gamma)}  \abs{\nabla  \psi_2}^2  dx-\varepsilon\int_{K_{12}(\gamma,\tilde \gamma)} \frac{\abs{P_\perp \psi_2}^2}{\abs{(x_1,x_2)}^2} dx \, .
    \end{split}
\end{equation} 

The functions $\psi_1$,  $\psi_2$ satisfy
\begin{equation} \label{ortho_once_more}
    P_0 \psi_1=0 \quad \text{and} \quad \abs{P_\perp \psi_1}^2 + \abs{P_\perp \psi_2}^2 = \abs{P_\perp \psi_{12}}^2 \, .
\end{equation}
Dividing by $\abs{x}^{2+ \tau}$ does not change the symmetry of the functions and consequently, together with \eqref{ortho_once_more} we get
\begin{equation} \label{orthogonality_used}
   \int_{K_{12}(\gamma)\setminus S(0,b)} \frac{\abs{\psi_{12}}^{2}}{\abs{x}^{2+\tau}} dx = \int_{K_{12}(\gamma)\setminus S(0,b)} \frac{\abs{ \psi_1}^{2}}{\abs{x}^{2+\tau}}dx + \int_{K_{12}(\gamma)\setminus S(0,b)} \frac{\abs{ \psi_2}^{2}}{\abs{x}^{2+\tau}} dx \, . 
\end{equation}
Inserting \eqref{IMS_and_ortho_first} and \eqref{orthogonality_used} into \eqref{l_12_copy_copy} yields
\begin{equation} \label{l_12_copy_next}
    L_{12}[\psi_{12}] \geq \tilde{L}_{12}[\psi_2] + \int_{K_{12}(\gamma)} \left( \abs{\nabla \psi_1}^2 - 2\frac{\abs{\psi_1}^{2}}{\abs{x}^{2+\tau}}  - \varepsilon \frac{\abs{\psi_1}^2}{\abs{x}^2}  \right) dx
\end{equation}
where $\tilde{L}_{12}[\psi_2]$ is given in \eqref{L_tilde_def}. 
Since $\supp{\psi_1} \subset \{x \in \R^3:\abs{x} \geq b\}$ we can apply the radial Hardy Inequality and for $\varepsilon>0$ small and $b>0$ large enough we obtain
\begin{equation}
    \int_{K_{12}(\gamma)} \left( \abs{\nabla \psi_1}^2 - 2\frac{\abs{\psi_1}^{2}}{\abs{x}^{2+\tau}}  - \varepsilon \frac{\abs{\psi_1}^2}{\abs{x}^2}  \right) dx \geq 0 \, . 
\end{equation}
This completes the proof of Lemma~\ref{first_part_of_lem21}. \qed
\subsubsection{Proof of Lemma~\ref{second_part_of_lem21}}  We aim to estimate 
    \begin{equation}
    \begin{split}
        \tilde{L}_{12}[\psi_2] = &\int_{K_{12}(\gamma)} \left( \abs{\nabla \psi_2}^2 + \sum_{\alpha\in I} V_{\alpha} \abs{\psi_2}^2 \mathbbm{1}_{K_{12}(\tilde \gamma)}  -2\frac{\abs{\psi_2}^{2}}{\abs{x}^{2+\tau}} \right)dx \\
        &-\varepsilon\int_{K_{12}(\gamma,\tilde \gamma)} \frac{\abs{P_\perp {\psi}_2}^2}{\abs{(x_1,x_2)}^2} dx \, .
    \end{split}
\end{equation}
Due to Lemma~\ref{potentials_small_outside_cones} there exists a constant $\tilde C>0$ such that on on $K_{12}(\gamma)\setminus S(0,b)$
\begin{equation} \label{v_23_13}
\abs{V_{13} + V_{23}} \leq \tilde C \abs{x}^{-2-\delta}  \, .
\end{equation}
Then, using \eqref{v_23_13} together with $\tau < \delta$ yields
\begin{equation} \label{tilde_12_step}
    \begin{split}
        \tilde{L}_{12}[\psi_2] \geq &\int_{K_{12}(\gamma)} \left( \abs{\nabla \psi_2}^2 +  V_{12} \abs{\psi_2}^2 \mathbbm{1}_{K_{12}(\tilde \gamma)}  -3\frac{\abs{\psi_2}^{2}}{\abs{x}^{2+\tau}} \right)dx \\
        &-\varepsilon\int_{K_{12}(\gamma,\tilde \gamma)} \frac{\abs{P_\perp \psi_2}^2}{\abs{(x_1,x_2)}^2} dx \, .
    \end{split}
\end{equation}
We rewrite \eqref{tilde_12_step} as \begin{equation} \label{tilde_12_step_2}
    \begin{split}
        \tilde{L}_{12}[\psi_2] \geq &\int_{K_{12}(\gamma)} \left( \abs{\nabla_{12} \psi_2}^2 +  V_{12} \abs{\psi_2}^2 \mathbbm{1}_{K_{12}(\tilde \gamma)} \right) dx \\
        & +\int_{K_{12}(\gamma)} \abs{\partial_{x_3} \psi_2}^2 dx -3\int_{K_{12}(\gamma)}\frac{\abs{\psi_2}^{2}}{\abs{x}^{2+\tau}} dx -\varepsilon\int_{K_{12}(\gamma,\tilde \gamma)} \frac{\abs{P_\perp \psi_2}^2}{\abs{(x_1,x_2)}^2} dx  \, .
    \end{split}
\end{equation}
We start with the first integral on the right--hand side of \eqref{tilde_12_step_2}. With the definition of $\tilde\psi_2$ in Section \ref{proof_of_L_1j_lemma} we have
\begin{equation} \label{obvious_grad_term}
    \int_{K_{12}(\gamma)} \abs{\nabla_{12} \psi_2}^2 dx =\int_{\R^3} \abs{\nabla_{12} \tilde \psi_2}^2 dx \, .
\end{equation}
The function $\tilde \psi_2$ and $\psi_2$ coincide inside of $K_{12}(\gamma)$ and consequently for the term involving $V_{12}$ in \eqref{tilde_12_step_2} we have
\begin{equation} \label{potential_cutted}
    \begin{split}
         \int_{K_{12}(\gamma)}V_{12} &\abs{\psi_2}^2 \mathbbm{1}_{K_{12}(\tilde \gamma)} dx \\
         &= \int_{\R^3} V_{12}\abs{\tilde \psi_2}^2 dx - \int_{K_{12}(\gamma, \tilde \gamma)} V_{12} \abs{\psi_2}^2 dx - \int_{\R^3 \setminus K_{12}(\gamma)} V_{12} \abs{\tilde \psi_2}^2 dx \, .
    \end{split}
\end{equation}
Outside of $K_{12}(\tilde \gamma)$ holds \mbox{$\abs{V_{12}}\leq \tilde C \abs{x}^{-3-\delta}$} for some $ \tilde C>0$. Consequently for $b>0$ large enough
\begin{equation} \label{potential_small_anulus}
    \abs{\int_{K_{12}(\gamma, \tilde \gamma)} V_{12} \abs{\psi_2}^2 dx} \leq \int_{K_{12}(\gamma, \tilde \gamma)} \abs{\psi_{2}(x)}^2 \abs{x}^{-2-\tau}   dx
\end{equation}
and
\begin{equation} \label{potential_small_outside_k12}
    \abs{ \int_{\R^3 \setminus K_{12}(\gamma)} V_{12} \abs{\tilde \psi_2}^2 dx } \leq   \int_{\R^3 \setminus K_{12}(\gamma)}  \abs{\tilde \psi_{2}(x)}^2 \abs{x}^{-2-\tau}  dx \, .
\end{equation}
Combining \eqref{obvious_grad_term}, \eqref{potential_cutted}, \eqref{potential_small_anulus} and \eqref{potential_small_outside_k12} we find
\begin{equation} \label{2d_energy_and_surface_int}
    \begin{split}
        \int_{K_{12}(\gamma)} \abs{\nabla_{12} \psi_2}^2 +  V_{12} \abs{\psi_{2}}^2  dx \geq &\int_{\R^3} \left( \abs{\nabla_{12} \tilde \psi_{2}}^2 +  V_{12} \abs{\tilde \psi_{2}}^2 \right) dx  \\
        &- \int_{\R^3 \setminus K_{12}(\gamma)}  \frac{\abs{\tilde \psi_{2}(x)}^2}{\abs{x}^{2+\tau}}  dx- \int_{K_{12}(\gamma, \tilde \gamma)} \frac{\abs{\psi_2}^2}{\abs{x}^{2+\tau}} dx \, .
    \end{split}
\end{equation}
By expressing $\tilde \psi_2$ in terms of $\Phi \varphi_0$ and $F$ (see \eqref{projection}), using $h_{12}\varphi_0=0$ and assertion (2) of Lemma~\ref{virtual_level_lemma} yields
\begin{equation} \label{there_is_mu}
    \int_{\R^3} \left( \abs{\nabla_{12} \tilde \psi_{2}}^2 +  V_{12} \abs{\tilde \psi_{2}}^2 \right) dx  \geq \mu \norm{ \nabla_{12} F}^2 \, .
\end{equation}
Inserting \eqref{there_is_mu} into \eqref{2d_energy_and_surface_int} gives
\begin{equation} \label{mu_and_surface_int}
    \begin{split}
        \int_{K_{12}(\gamma)} &\abs{\nabla_{12} \psi_2}^2 +  V_{12} \abs{\psi_{2}}^2  dx \\
        &\geq \mu \norm{ \nabla_{12} F}^2 - \int_{\R^3 \setminus K_{12}(\gamma)}  \abs{\tilde \psi_{2}(x)}^2 \abs{x}^{-2-\tau}  dx- \int_{K_{12}(\gamma, \tilde \gamma)} \abs{\psi_2}^2\abs{x}^{-2-\tau} dx \, .
    \end{split}
\end{equation}

Next, we show that 
\begin{equation}
    \int_{\R^3 \setminus K_{12}(\gamma)}  \abs{\tilde \psi_{2}(x)}^2 \abs{x}^{-2-\tau}  dx
\end{equation}
can be estimated by an integral over the surface $\partial K_{12}(\gamma)$.  We introduce polar coordinates
\begin{equation} \label{polar_for_12}
    \begin{split}
        x_1&=\rho  \sin(\varphi) , \\
        x_2&=\rho  \cos(\varphi) 
    \end{split}
\end{equation}
with $\varphi\in[0,2\pi)$ and $\rho \in[0,\infty)$. In this choice of coordinates the set $K_{12}(\gamma)$ is determined by $\rho \leq \kappa \abs{x_3}$ for some $\kappa>0$ depending on $\gamma$ (see Figure \ref{fig:002}). Then
\begin{equation} \label{get_surface_int_before}
    \begin{split}
        \int_{\R^3 \setminus K_{12}(\gamma)}  \frac{\abs{\tilde \psi_{2}(x)}^2}{\abs{x}^{2+\tau}}  dx &=  \int_{-\infty}^\infty \int_{\kappa \abs{x_3}}^\infty \frac{\int_{0}^{2\pi} \abs{\tilde \psi_{2}(\rho,\varphi, x_3)}^2d\varphi }{\rho^{2+\tau}}  \rho d\rho dx_3 \, .
    \end{split}
\end{equation}
Outside of the conical set $K_{12}(\gamma)$ the function $\tilde \psi_{2}$ is equal to its value on the boundary and consequently substituting $\tilde \psi_{2}(\rho,\varphi, x_3)$ with $\tilde \psi_{2}(\kappa \abs{x_3}, \varphi, x_3)$ in \eqref{get_surface_int} and solving the integral over $\rho$ yields
\begin{equation} \label{get_surface_int}
    \begin{split}
        \int_{\R^3 \setminus K_{12}(\gamma)}  \frac{\abs{\tilde \psi_{2}(x)}^2}{\abs{x}^{2+\tau}}  dx &= \int_{-\infty}^\infty \int_{\kappa \abs{x_3}}^\infty \frac{\int_{0}^{2\pi} \abs{\tilde \psi_{2}(\kappa \abs{x_3} ,\varphi, x_3)}^2d\varphi }{\rho^{2+\tau}}  \rho d\rho dx_3 \\
        &=(\kappa^{\tau}\tau)^{-1}\int_{-\infty}^\infty \frac{\int_{0}^{2\pi} \abs{\tilde \psi_{2}(\kappa \abs{x_3} ,\varphi, x_3)}^2d\varphi }{\abs{x_3}^{1+\tau}} \abs{x_3} dx_3 \, . 
    \end{split}
\end{equation}
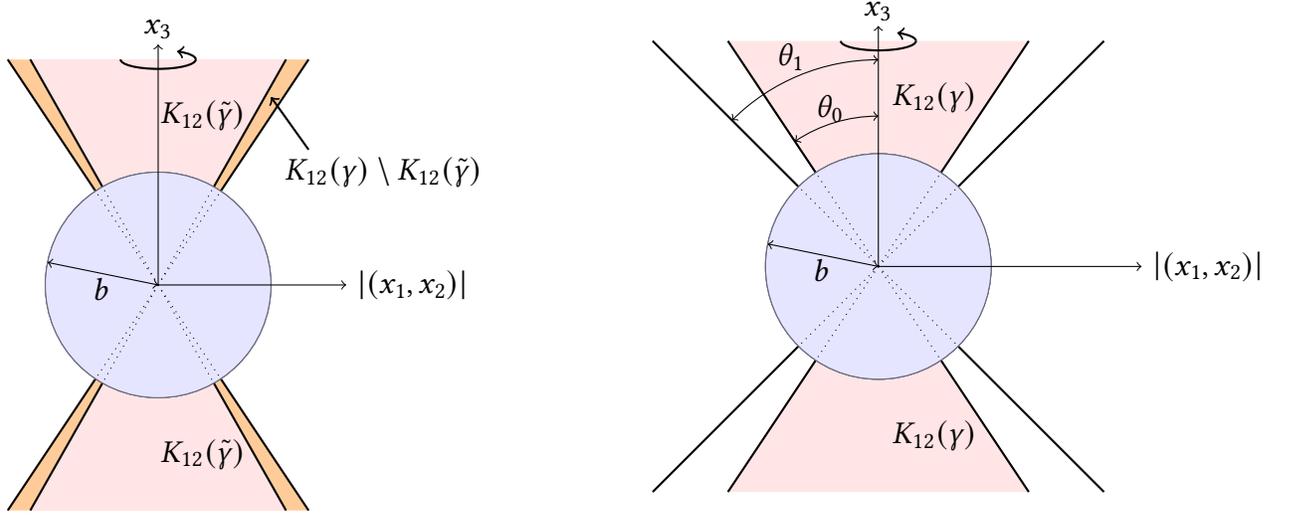
\begin{figure}[!ht]
\begin{center}
\begin{minipage}{.45\textwidth}
\begin{tikzpicture}[scale=1] 
    \fill[orange,opacity=0.4] (0,0) -- (2,-3) -- (-2,-3) -- cycle;
    \fill[orange,opacity=0.4] (0,0) -- (2,3) -- (-2,3) -- cycle;

    \fill[white] (0,0) -- (1.7,3) -- (-1.7,3) -- cycle;
    \fill[white] (0,0) -- (1.7,-3) -- (-1.7,-3) -- cycle;
    \fill[red!20,opacity=0.5] (0,0) -- (1.7,3) -- (-1.7,3) -- cycle;
    \fill[red!20,opacity=0.5] (0,0) -- (1.7,-3) -- (-1.7,-3) -- cycle;

    \draw [->, thick] (2,1.8) -- (1.5,2.5);
    \node at (3,1.5) {$K_{12}(\gamma) \setminus K_{12}(\tilde \gamma)$};
    
    \draw[black] (0,0) circle (1.5cm);
    \draw[fill, blue!20,opacity=0.5] (0,0) circle (1.5cm);
    \draw [thick] (0,0) -- (-2,3);
    \draw [thick] (0,0) -- (2,3);
    \draw [thick] (0,0) -- (-2,-3);
    \draw [thick] (0,0) -- (2,-3);

    \draw [thick] (0,0) -- (-1.7,3);
    \draw [thick] (0,0) -- (1.7,3);
    \draw [thick] (0,0) -- (-1.7,-3);
    \draw [thick] (0,0) -- (1.7,-3);

    \draw[fill, white] (0,0) circle (1.5cm);
    \draw[black] (0,0) circle (1.5cm);
    \draw[fill, blue!20,opacity=0.5] (0,0) circle (1.5cm);

    \draw [ ->] (0,0) -- (0,3.2);
    \node at (0,3.4) {$x_3$};
    \draw[->] (0,0) -- (2.5,0) node[anchor=west] {$\abs{(x_1,x_2)}$};
    \draw [dotted] (0,0) -- (-1.7,3);
    \draw [dotted] (0,0) -- (1.7,3);
    \draw [dotted] (0,0) -- (-1.7,-3);
    \draw [dotted] (0,0) -- (1.7,-3);
    \draw [dotted] (0,0) -- (-2,3);
    \draw [dotted] (0,0) -- (2,3);
    \draw [dotted] (0,0) -- (-2,-3);
    \draw [dotted] (0,0) -- (2,-3);

    \draw [<->] (0,0) -- (-1.47,0.3);
    \node at (-0.75,-0.05) {$b$};

    \node at (0.6,2.25) {$K_{12}(\tilde\gamma)$};
    \node at (0.6,-2.25) {$K_{12}(\tilde\gamma)$};
    
\draw [thick] (-0.5,3,0) arc[start angle=180, end angle=360, x radius=0.5, y radius=0.125];
\draw [->,thick] (0.5,3,0) arc[start angle=0, end angle=60, x radius=0.5, y radius=0.125];
\end{tikzpicture}
\end{minipage}
\hfill
\begin{minipage}{.45\textwidth}
\begin{tikzpicture}[scale=1] 
    \fill[red!20,opacity=0.5] (0,0) -- (2,3) -- (-2,3) -- cycle;
    \fill[red!20,opacity=0.5] (0,0) -- (2,-3) -- (-2,-3) -- cycle;
    \draw[black] (0,0) circle (1.5cm);
    \draw[fill, blue!20,opacity=0.5] (0,0) circle (1.5cm);
    \draw [thick] (0,0) -- (-2,3);
    \draw [thick] (0,0) -- (2,3);
    \draw [thick] (0,0) -- (-2,-3);
    \draw [thick] (0,0) -- (2,-3);

    \draw [thick] (0,0) -- (-3,3);
    \draw [thick] (0,0) -- (3,3);
    \draw [thick] (0,0) -- (-3,-3);
    \draw [thick] (0,0) -- (3,-3);

    \draw
    (0,1.5) coordinate (a) 
    -- (0,0) coordinate (b) 
    -- (-1.5,2.25) coordinate (c) 
    pic["$\theta_0$", draw=black, <->, angle eccentricity=1.1, angle radius=2cm]
    {angle=a--b--c};

    \draw
    (0,1.5) coordinate (a) 
    -- (0,0) coordinate (b) 
    -- (-3,3) coordinate (c) 
    pic["$\theta_1$", draw=black, <->, angle eccentricity=1.1, angle radius=2.75cm]
    {angle=a--b--c};
    
    \draw[fill, white] (0,0) circle (1.5cm);
    \draw[black] (0,0) circle (1.5cm);
    \draw[fill, blue!20,opacity=0.5] (0,0) circle (1.5cm);

    \draw [ ->] (0,0) -- (0,3.2);
    \node at (0,3.4) {$x_3$};
    \draw[->] (0,0) -- (3.5,0) node[anchor=west] {$\abs{(x_1,x_2)}$};
    \draw [dotted] (0,0) -- (-3,3);
    \draw [dotted] (0,0) -- (3,3);
    \draw [dotted] (0,0) -- (-3,-3);
    \draw [dotted] (0,0) -- (3,-3);
    \draw [dotted] (0,0) -- (-2,3);
    \draw [dotted] (0,0) -- (2,3);
    \draw [dotted] (0,0) -- (-2,-3);
    \draw [dotted] (0,0) -- (2,-3);

    \draw [<->] (0,0) -- (-1.47,0.3);
    \node at (-0.75,-0.05) {$b$};

    \node at (0.75,2.25) {$K_{12}(\gamma)$};
    \node at (0.75,-2.25) {$K_{12}(\gamma)$};
    
\draw [thick] (-0.5,3,0) arc[start angle=180, end angle=360, x radius=0.5, y radius=0.125];
\draw [->,thick] (0.5,3,0) arc[start angle=0, end angle=60, x radius=0.5, y radius=0.125];
\end{tikzpicture} 
\end{minipage}
\end{center}
\caption{Left--hand side: sketch of the sets $K_{12}(\gamma)$ and $K_{12}(\tilde \gamma)$ used in the proof of Lemma~\ref{L_1j_lemma}. \\Right--hand side: sketch of the sets $K_{12}(\gamma)$ and $K_{12}(\gamma_1)$ where the angles $\theta_0$ and $\theta_1$ are defined as $\theta_0=\arcsin(\gamma)$ and $\theta_1=\arcsin (\gamma_1)$ and used in Lemma~\ref{trace_lemma_001}. } \label{fig:002}
\end{figure}
Regarding the set $\partial K_{23}(\gamma)$ the surface measure $d\sigma$ equals $\abs{x_3}dx_3 d\varphi$ up to a constant depending on $\gamma$ and the function $P_\perp \tilde \psi_2 = 0$ such that $\tilde \psi_2 = P_0 \psi$ on this surface. For $(x_1,x_2,x_3) \in \partial K_{12}(\gamma)$ we have $\abs{x_3} = (1-\gamma^2)^{1/2} \abs{x}$ and therefore there exists some $C_1>0$ that depends on $\gamma$ and $\delta$ but is independent of $\psi$ such that
\begin{equation} \label{est_with_surf_int}
    \int_{\R^3 \setminus K_{12}(\gamma)}  \frac{\abs{\tilde \psi_{2}(x)}^2}{\abs{x}^{2+\tau}}  dx \leq C_1 \int_{\partial K_{12}(\gamma) } \frac{\abs{P_0 \psi }^2}{\abs{x}^{1+\tau}} d\sigma \, .
\end{equation}

Combining \eqref{est_with_surf_int} and \eqref{mu_and_surface_int} we arrive at
\begin{equation} \label{estimate_first_ingral_grad12}
    \begin{split}
    \int_{K_{12}(\gamma)} \abs{\nabla_{12} \psi_2}^2 +  V_{12} \abs{\psi_{2}}^2  dx  
    \geq \mu \norm{ \nabla_{12} F}^2 &- C_1\int_{\partial K_{12}(\gamma) } \frac{\abs{(P_0\psi )(x)}^2}{\abs{x}^{1+\tau}} d\sigma \\
    &- \int_{K_{12}(\gamma, \tilde \gamma)} \frac{\abs{\psi_2}^2}{\abs{x}^{2+\tau}} dx \, .
    \end{split}
\end{equation}
Substituting \eqref{estimate_first_ingral_grad12} into \eqref{tilde_12_step_2} yields
\begin{equation} \label{tilde_12_step_2_1}
    \begin{split}
        \tilde{L}_{12}[\psi_2] \geq &\mu \norm{ \nabla_{12} F}^2 - C_1\int_{\partial K_{12}(\gamma) } \frac{\abs{P_0 \psi }^2}{\abs{x}^{1+\tau}} d\sigma  \\
        & +\int_{K_{12}(\gamma)} \left( \abs{\partial_{x_3} \psi_2}^2  -4\frac{\abs{\psi_2}^{2}}{\abs{x}^{2+\tau}} \right) dx -\varepsilon\int_{K_{12}(\gamma,\tilde \gamma)} \frac{\abs{P_\perp \psi_2}^2}{\abs{(x_1,x_2)}^2} dx  \, .
    \end{split}
\end{equation}
We proceed by studying the term
\begin{equation} \label{split_up_psi_2}
    \begin{split}
        \int_{K_{12}(\gamma)} \left( \abs{\partial_{x_3} \psi_2}^2 dx -4\frac{\abs{\psi_2}^{2}}{\abs{x}^{2+\tau}} \right) dx = &\int_{K_{12}(\gamma)} \left( \abs{\partial_{x_3}  (P_0 \psi_2)}^2 -4\frac{\abs{P_0 \psi_2}^{2}}{\abs{x}^{2+\tau}} \right) dx \\
        &+ \int_{K_{12}(\gamma)} \left( \frac{1}{2}\abs{ \partial_{x_3} (P_\perp \psi_2)}^2  -4\frac{\abs{P_\perp \psi_2}^{2}}{\abs{x}^{2+\tau}} \right) dx   \\
        &+\int_{K_{12}(\gamma)}\frac{1}{2} \abs{ \partial_{x_3} (P_\perp \psi_2)}^2 dx 
    \end{split}
\end{equation}
Using $P_\perp \psi_2 =0$  on $\partial K_{12}(\gamma)$ and decreasing the integral by replacing $\abs{x}$ with $\abs{x_3}$ together with the one--dimensional Hardy Inequality (see Lemma~\ref{one_d_hardy}) yields for $b>0$ large enough
\begin{equation} \label{one_dim_hard_again}
 \int_{K_{12}(\gamma)} \left( \frac{1}{2}\abs{\partial_{x_3} (P_\perp \psi_2)}^2  -4\frac{\abs{P_\perp \psi_2}^{2}}{\abs{x_3}^{2+\tau}} \right) dx   \geq 0 \, .
\end{equation}
Combining \eqref{split_up_psi_2} and \eqref{one_dim_hard_again} yields
\begin{equation} \label{only_p0_terms_left}
    \begin{split}
        \int_{K_{12}(\gamma)} \left( \abs{\partial_{x_3} \psi_2}^2 dx -4\frac{\abs{\psi_2}^{2}}{\abs{x}^{2+\tau}} \right) dx  \geq 
        & \int_{K_{12}(\gamma)} \left( \abs{\partial_{x_3}  (P_0 \psi_2)}^2  -4\frac{\abs{P_0 \psi_2(x)}^{2}}{\abs{x_3}^{2+\tau}} \right) dx \\
        &+\frac{1}{2} \int_{K_{12}(\gamma)}  \abs{\partial_{x_3} (P_\perp \psi_2)}^2 dx \, .
    \end{split}
\end{equation}
Next, we estimate the integral involving $P_0 \psi_2 $ in \eqref{only_p0_terms_left}. The one--dimensional Hardy Inequality can not be applied directly, as $P_0 \psi_2$ does not vanish on the boundary $\partial K_{12}(\gamma)$. So, we use the following construction instead. 

Let $G$ be defined as a continuous function in $K_{12}(\gamma) \setminus S(0,b)$ that coincides with $P_0 \psi_2$ on the boundary $\partial K_{12}(\gamma)$ and is independent of $x_3$ within $K_{12}(\gamma)$. We define $\Gamma$ in $K_{12}(\gamma)\setminus S(0,b)$ as
\begin{equation}
    \Gamma \coloneqq P_0 \psi_2 - G 
\end{equation}
such that $\Gamma$ vanishes on $\partial K_{12}(\gamma)$. Then
\begin{equation} \label{obvious_with_Gamma}
    \int_{K_{12}(\gamma)} \abs{\partial_{x_3}  (P_0 \psi_2)}^2 dx =  \int_{K_{12}(\gamma)} \abs{\partial_{x_3}  \Gamma }^2 dx
\end{equation}
and 
\begin{equation} \label{splitted_Gamma_G}
    \begin{split}
    \int_{K_{12}(\gamma)} \frac{\abs{P_0 \psi_2}^{2}}{\abs{x_3}^{2+\tau}} dx  &= \int_{K_{12}(\gamma)} \frac{\abs{\Gamma + G}^{2}}{\abs{x_3}^{2+\tau}} dx \\
    &\leq 2\int_{K_{12}(\gamma)} \frac{\abs{\Gamma}^{2}}{\abs{x_3}^{2+\tau}} dx  + 2\int_{K_{12}(\gamma)} \frac{\abs{G}^{2}}{\abs{x_3}^{2+\tau}} dx \, .
    \end{split}
\end{equation}
Combining \eqref{obvious_with_Gamma} and \eqref{splitted_Gamma_G} we find
\begin{equation}
    \begin{split}
        \int_{K_{12}(\gamma)} \left( \abs{\partial_{x_3}  (P_0 \psi_2)}^2  -4\frac{\abs{P_0 \psi_2(x)}^{2}}{\abs{x_3}^{2+\tau}} \right) dx \geq  &\int_{K_{12}(\gamma)} \left( \abs{\partial_{x_3}  \Gamma }^2  -8\frac{\abs{\Gamma}^{2}}{\abs{x_3}^{2+\tau}} \right) dx \\
        &-\int_{K_{12}(\gamma)}8\frac{\abs{G}^{2}}{\abs{x_3}^{2+\tau}}dx
    \end{split}
\end{equation}
Since $\psi$ and consequently $\Gamma$ vanishes for $\abs{x_3}<b/2 $ 
we can apply the one--dimensional Hardy inequality (see Lemma~\ref{one_dimensional_hardy}), which yields for $\tau>0$ and $b>0$ large enough
\begin{equation} \label{using_1d_hardy_on_Gamma}
    \int_{K_{12}(\gamma)} \abs{\partial_{x_3}  \Gamma }^2 dx - 8\int_{K_{12}(\gamma)} \frac{\abs{\Gamma}^{2}}{\abs{x_3}^{2+\tau}} dx   \geq 0 \, .
\end{equation}
This yields
\begin{equation} \label{another_surface_int}
  \int_{K_{12}(\gamma)} \left( \abs{\partial_{x_3}  (P_0 \psi_2)}^2 dx -4\frac{\abs{P_0 \psi_2(x)}^{2}}{\abs{x}^{2+\tau}} \right) dx \geq  -8\int_{K_{12}(\gamma)} \frac{\abs{G}^{2}}{\abs{x_3}^{2+\tau}} dx \, .
\end{equation}
Next, we show that the integral on the right--hand side of equation \eqref{another_surface_int} can be estimated by an integral over $\partial K_{12}(\gamma)$. The function $G$ is independent of $x_3$ in $K_{12}(\gamma)$, therefore using polar coordinates as in \eqref{polar_for_12} we find
\begin{equation} 
    \begin{split}
        \int_{K_{12}(\gamma)} \frac{\abs{G(x)}^{2}}{\abs{x_3}^{2+\tau}} dx  &=  \int_{0}^\infty \int_{\abs{x_3}\geq \kappa^{-1} \rho} \frac{\int_{0}^{2\pi} \abs{ G(\rho,\varphi, x_3)}^2d\varphi }{\abs{x_3}^{2+\tau}} dx_3  \rho d\rho \\
        &=\int_{0}^\infty \int_{\abs{x_3}\geq \kappa^{-1} \rho} \frac{\int_{0}^{2\pi} \abs{ G(\rho,\varphi, \kappa^{-1} \rho)}^2d\varphi }{\abs{x_3}^{2+\tau}} dx_3  \rho d\rho \\
        &=(\kappa^{1+\tau}(1+\tau))^{-1} \int_{0}^\infty  \frac{\int_{0}^{2\pi} \abs{ G(\rho,\varphi, \kappa^{-1} \rho)}^2d\varphi }{\rho^{1+\tau}}  \rho d\rho 
    \end{split}
\end{equation}
Due to the definition of $G$ and since $\rho = \gamma \abs{x}$ on $\partial K_{12}(x)$ there exists a constant $C_2>0$ that depends on $\gamma$ and $\delta$ but is independent of $\psi$ such that
\begin{equation}\label{surface_int_with_tau}
    \int_{K_{12}(\gamma)} \frac{\abs{G}^{2}}{\abs{x_3}^{2+\tau}} dx = C_2 \int_{\partial K_{12}(\gamma) } \frac{\abs{P_0\psi }^2}{\abs{x}^{1+\tau}} d\sigma \, .
\end{equation}
Substituting the relation \eqref{surface_int_with_tau} into \eqref{another_surface_int} it follows from \eqref{only_p0_terms_left} that
\begin{equation} \label{second_part_in_proof_of_64}
    \begin{split}
        \int_{K_{12}(\gamma)} \left( \abs{\partial_{x_3} \psi_2}^2 dx -4\frac{\abs{\psi_2}^{2}}{\abs{x}^{2+\tau}} \right) dx  \geq 
        & -8C_2 \int_{\partial K_{12}(\gamma) } \frac{\abs{P_0\psi }^2}{\abs{x}^{1+\tau}} d\sigma  \\
        &+\frac{1}{2} \int_{K_{12}(\gamma)}  \abs{\partial_{x_3} (P_\perp \psi_2)}^2 dx \, .
    \end{split}
\end{equation}
We insert \eqref{second_part_in_proof_of_64} into \eqref{tilde_12_step_2_1} and define $C \coloneqq C_1+8C_2$, such that
\begin{equation}
    \begin{split}
        \tilde{L}_{12}[\psi_2] \geq &\mu \norm{ \nabla_{12} F}^2  + \frac{1}{2}\int_{K_{12}(\gamma)} \abs{\partial_{x_3} (P_\perp \psi_2)}^2 dx -\varepsilon\int_{K_{12}(\gamma,\tilde \gamma)} \frac{\abs{P_\perp {\psi}_2}^2}{\abs{(x_1,x_2)}^2} dx \\
        &- C \int_{\partial K_{\alpha}(\gamma) } \frac{\abs{P_0 \psi }^2}{\abs{x}^{1+\tau}} d\sigma,
    \end{split}
\end{equation}
which completes the proof of Lemma~\ref{second_part_of_lem21}.  \qed
\subsubsection{Proof of Lemma~\ref{third_part_of_lem21}} \label{prof_of_lem_53}  To prove the lemma it suffices to show that for any $\mu>0$ and \mbox{$\varepsilon\in (0,\mu/8)$} there exists a $\lambda \in (0,1/2)$ such that for all $b>0$ (depending on $\mu, \varepsilon, \lambda$) large enough, the following inequality holds:
\begin{equation} \label{remains_for_21_copy}
    \mu \norm{ \nabla_{12} F}^2  + \lambda \int_{K_{12}(\gamma)} \abs{\partial_{x_3} (P_\perp \psi_2)}^2 dx -\varepsilon\int_{K_{12}(\gamma,\tilde \gamma)} \frac{\abs{P_\perp {\psi}_2}^2}{\abs{(x_1,x_2)}^2} dx \geq 0 \, .
\end{equation}
We start with the second term on the left--hand side of \eqref{remains_for_21_copy}. 
The function $P_\perp \psi_2$ vanishes for $\abs{x_3}=0$ and therefore by the one--dimensional Hardy Inequality (Lemma~\ref{one_d_hardy})
\begin{equation} \label{next_step001}
    \begin{split}
             \int_{K_{12}(\gamma)} \abs{\partial_{x_3} (P_\perp \psi_2)}^2 dx  &\geq \frac{1}{4} \int_{K_{12}(\gamma)} \frac{\abs{ P_\perp \psi_2}^2}{\abs{x_3}^2} dx \, . 
    \end{split}
\end{equation}
Since $\psi_2 = \Phi \varphi_0 + F$ on $K_{12}(\gamma)$ and  $(a+b)^2 \geq a^2/2 - b^2$ we can estimate the right--hand side of \eqref{next_step001} by
\begin{equation}\label{next_step002}
    \begin{split}
             \int_{K_{12}(\gamma)} \frac{\abs{ P_\perp \psi_2}^2}{\abs{x_3}^2} dx  &\geq \frac{1}{2} \int_{K_{12}(\gamma)} \abs{\Phi}^2 \frac{\abs{ P_\perp \varphi_0}^2}{\abs{x_3}^2} dx -\int_{K_{12}(\gamma)} \frac{\abs{ P_\perp F}^2}{\abs{x_3}^2} dx \, .
    \end{split}
\end{equation}
Using that 
\begin{equation}
    \abs{(x_1,x_2)}^2\leq \frac{\gamma^2}{1-\gamma^2} \, x_3^2, \quad \forall x\in K_{12}(\gamma).
\end{equation}
and substituting this into the right--hand side of \eqref{next_step002} we find
\begin{equation}\label{next_step002_2}
    \begin{split}
             \int_{K_{12}(\gamma)} \frac{\abs{ P_\perp \psi_2}^2}{\abs{x_3}^2} dx  \geq \frac{1}{2} \int_{K_{12}(\gamma)} \abs{\Phi}^2 \frac{\abs{ P_\perp \varphi_0}^2}{\abs{x_3}^2} dx - \frac{\gamma^2}{1-\gamma^2}\int_{K_{12}(\gamma)} \frac{\abs{ P_\perp F}^2}{\abs{(x_1,x_2)}^2} dx \, .
    \end{split}
\end{equation}
Combining \eqref{next_step001} and \eqref{next_step002_2}  and assuming 
\begin{equation}
    \lambda < \frac{1-\gamma^2}{2\gamma^2} \mu
\end{equation}
we arrive at
\begin{equation}\label{intermed_step}
    \begin{split}
        \lambda \int_{K_{12}(\gamma)} \abs{\partial_{x_3} (P_\perp \psi_2)}^2 dx \geq  &\frac{\lambda}{8}\int_{K_{12}(\gamma)} \abs{\Phi}^2 \frac{\abs{ P_\perp \varphi_0}^2}{\abs{x_3}^2} dx - \frac{\lambda}{4}\frac{\gamma^2}{1-\gamma^2}\int_{K_{12}(\gamma)} \frac{\abs{ P_\perp F}^2}{\abs{(x_1,x_2)}^2} dx \\
        &\geq  \frac{\lambda}{8}\int_{K_{12}(\gamma)} \abs{\Phi}^2 \frac{\abs{ P_\perp \varphi_0}^2}{\abs{x_3}^2} dx - \frac{\mu}{8}\int_{K_{12}(\gamma)} \frac{\abs{ P_\perp F}^2}{\abs{(x_1,x_2)}^2} dx
    \end{split}
\end{equation}
Using \mbox{$\varepsilon\in (0,\mu/8)$} we find that last term on the left--hand side of \eqref{remains_for_21_copy} can be estimated as
\begin{equation} \label{obvious_next_step}
\begin{split}
    \varepsilon \int_{K_{12}(\gamma,\tilde \gamma)} \frac{\abs{P_\perp {\psi}_2}^2}{\abs{(x_1,x_2)}^2} dx &\leq 2\varepsilon \int_{K_{12}(\gamma,\tilde \gamma)} \abs{\Phi}^2 \frac{\abs{P_\perp \varphi_0}^2}{\abs{(x_1,x_2)}^2} dx +2 \varepsilon \int_{K_{12}(\gamma,\tilde \gamma)} \frac{\abs{P_\perp F}^2}{\abs{(x_1,x_2)}^2} dx \\
    &\leq  2\varepsilon \int_{K_{12}(\gamma,\tilde \gamma)} \abs{\Phi}^2 \frac{\abs{P_\perp \varphi_0}^2}{\abs{(x_1,x_2)}^2} dx + \frac{\mu}{4} \int_{K_{12}(\gamma,\tilde \gamma)} \frac{\abs{P_\perp F}^2}{\abs{(x_1,x_2)}^2} dx \, .
\end{split}
\end{equation}
Inserting \eqref{intermed_step} and \eqref{obvious_next_step} into \eqref{remains_for_21_copy} we find 
\begin{equation} \label{reoder_after_in_lambda}
    \begin{split}
         \mu \norm{ \nabla_{12} F}^2 + \lambda \int_{K_{12}(\gamma)} &\abs{\partial_{x_3} (P_\perp \psi_2)}^2 dx -\varepsilon\int_{K_{12}(\gamma,\tilde \gamma)} \frac{\abs{P_\perp {\psi}_2}^2}{\abs{(x_1,x_2)}^2} dx\\
        &\geq \frac{\lambda}{8} \int_{K_{12}(\gamma)} \abs{\Phi}^2 \frac{\abs{ P_\perp \varphi_0}^2}{\abs{x_3}^2} dx- 2\varepsilon\int_{K_{12}(\gamma,\tilde \gamma)} \abs{\Phi}^2 \frac{\abs{P_\perp \varphi_0}^2}{\abs{(x_1,x_2)}^2} dx \\
        & + \mu \norm{ \nabla_{12} F}^2 - \frac{3}{8}\mu \int_{K_{12}(\gamma)}\frac{\abs{P_\perp F}^2}{\abs{(x_1,x_2)}^2} dx \, .
    \end{split}
\end{equation}
Furthermore, due to the symmetry of $P_\perp F$ (see \eqref{Hardy_type_ineq}) we have
\begin{equation}
    \norm{ \nabla_{12} F}^2 = \norm{ \nabla_{12} (P_0F)}^2+\norm{ \nabla_{12} (P_\perp F)}^2 \geq \norm{ \nabla_{12} (P_\perp F)}^2 \geq 4 \int_{\R^3}\frac{\abs{P_\perp F}^2}{\abs{(x_1,x_2)}^2} dx \, .
\end{equation}
Consequently for the terms in that last line of \eqref{reoder_after_in_lambda} we find
\begin{equation}
    \mu \norm{ \nabla_{12} F}^2 - \frac{3}{8}\mu \int_{K_{12}(\gamma)}\frac{\abs{P_\perp F}^2}{\abs{(x_1,x_2)}^2} dx \geq 0 \, .
\end{equation}
To complete the proof of the lemma it remains to show for fixed $\lambda$ and $\varepsilon$ we can choose $b>0$ large enough such that
\begin{equation} \label{remains_for_5.3}
    \frac{\lambda}{8} \int_{K_{12}(\gamma)} \abs{\Phi}^2 \frac{\abs{ P_\perp \varphi_0}^2}{\abs{x_3}^2} dx- 2\varepsilon\int_{K_{12}(\gamma,\tilde \gamma)} \abs{\Phi}^2 \frac{\abs{P_\perp \varphi_0}^2}{\abs{(x_1,x_2)}^2} dx \geq 0 \, .
\end{equation}
The first integral in \eqref{remains_for_5.3} is taken over the region $K_{12}(\gamma)$, while the second integral, which is negative, is taken over $K_{12}(\gamma, \tilde \gamma)$, a subset of $K_{12}(\gamma)$. We will show that \eqref{remains_for_5.3} follows from this observation and the decay properties of $P_\perp \varphi_0$ proved in Lemma~\ref{first_lemma}.

It holds
\begin{equation} \label{geometry_part_1}
    \frac{\tilde \gamma^2}{1- \tilde \gamma^2} \, x_3^2 \leq \abs{(x_1,x_2)}^2, \quad \forall x\in K_{12}(\gamma,\tilde \gamma),
\end{equation}
and
\begin{equation} \label{geometry_part_2}
    \tilde \gamma  b/2 \leq \tilde \gamma \abs{x} \leq \abs{(x_1,x_2)}, \quad \forall x\in K_{12}(\gamma,\tilde \gamma) \cap \supp(\psi) \, .
\end{equation}
Using \eqref{geometry_part_1} and \eqref{geometry_part_2} and applying Lemma~\ref{first_lemma} there exists some $\nu>0$ and a constant $c(\tilde \gamma,\nu)>0$ such that for the second integral in \eqref{remains_for_5.3}
\begin{equation} \label{next_step005}
    \begin{split}
        &\int_{K_{12}(\gamma,\tilde \gamma)} \abs{\Phi}^2 \frac{\abs{P_\perp \varphi_0}^2}{\abs{(x_1,x_2)}^2} dx \\
        &=\int_{K_{12}(\gamma,\tilde \gamma)} \abs{\Phi}^2 \frac{\abs{P_\perp \varphi_0}^2}{\abs{(x_1,x_2)}^2} \frac{\left( 1+\abs{(x_1,x_2)} \right)^\nu}{\left( 1+\abs{(x_1,x_2)} \right)^\nu}dx \\
        &\leq \frac{1}{\left(1+\tilde \gamma b/2\right)^\nu} \frac{1-\tilde \gamma^2}{\tilde \gamma^2} \int_{K_{12}(\gamma,\tilde \gamma)} \frac{\abs{\Phi}^2}{\abs{x_3}^2} \left( 1+\abs{(x_1,x_2)} \right)^\nu \abs{P_\perp \varphi_0}^2 d(x_1,x_2) dx_3 \\
        &= \frac{c(\tilde \gamma,\nu)}{b^\nu} \int_{b/2}^\infty \frac{\abs{\Phi}^2}{\abs{x_3}^2} dx_3 \, .
    \end{split}
\end{equation}
On the other hand for the first term in \eqref{remains_for_5.3} we have for $b>0$ large enough
\begin{equation} \label{next_step006}
    \int_{K_{12}(\gamma)} \abs{\Phi}^2 \frac{\abs{ P_\perp \varphi_0}^2}{\abs{x_3}^2} dx  \geq \frac{\norm{P_\perp \varphi_0}_{L^2(\R^2)}^2}{2} \int_{b/2}^\infty \frac{\abs{\Phi}^2}{\abs{x_3}^2} dx_3 \, .
\end{equation}
Combining \eqref{next_step005} and \eqref{next_step006} proves \eqref{remains_for_5.3}, which completes the proof of Lemma~\ref{third_part_of_lem21} and as discussed in \ref{proof_of_L_1j_lemma} this completes also the proof of Lemma~\ref{L_1j_lemma}. \qed 
\subsection{Proof of Lemma~\ref{L_23_lemma}} \label{proof_of_L_23_lemma}
We aim to  estimate 
\begin{equation}
    L_{23}[\psi_{23}]  = \int_{K_{23}(\gamma)} \left( \abs{\nabla \psi_{23}}^2 + \sum_{\beta\in I} V_{\beta} \abs{\psi_{23}}^2 \right) dx -  \int_{K_{23}(\gamma)\setminus S(0,b)} \frac{\abs{\psi_{23}}^{2}}{\abs{x}^{2+\tau}} dx \, .
\end{equation}
The potentials $V_{1j}$ satisfy $\abs{V_{1j}}\leq \tilde C \abs{x}^{-3-\delta}$ for $j\in \{2,3\}$ and some $\tilde C>0$ on $K_{23}(\gamma)$ due to Lemma~\ref{potentials_small_outside_cones}. Consequently, for $b>0$ sufficiently large
\begin{equation} \label{L_23_got_rid_of_pot}
    L_{23}[\psi_{23}]  \geq \int_{K_{23}(\gamma)} \left( \abs{\nabla \psi_{23}}^2 + V_{23} \abs{\psi_{23}}^2 \right) dx -  2 \int_{K_{23}(\gamma)\setminus S(0,b)} \frac{\abs{\psi_{23}}^{2}}{\abs{x}^{2+\tau}} dx,
\end{equation} 
which is equivalent to
\begin{equation} \label{idendify_2part_op}
    \begin{split}
        L_{23}[\psi_{23}]  \geq &\int_{K_{23}(\gamma)} \left( \abs{\partial_{x_2} \psi_{23}}^2 + \abs{\partial_{x_3} \psi_{23}}^2 + V_{23} \abs{\psi_{23}}^2 \right) dx \\
        &+\int_{K_{23}(\gamma)} \abs{\partial_{x_1} \psi_{23}}^2  -  2 \int_{K_{23}(\gamma)\setminus S(0,b)} \frac{\abs{\psi_{23}}^{2}}{\abs{x}^{2+\tau}} dx \, .
    \end{split}
\end{equation}
Next, we estimate the first term on the right--hand side of equation \eqref{idendify_2part_op}, which corresponds to a part of the quadratic form of the operator $h_{23}$ defined in \eqref{two_particle_operators}. Note that the operator $h_{23}$ is translation invariant. We introduce new coordinates $(q,\xi)$ which correspond to the relative distance and position of the center of mass of the subsystem $(23)$ with 
\begin{equation}
    q \coloneqq \frac{1}{\sqrt{M_{23}}}(\sqrt{m_3}x_2 - \sqrt{m_2}x_3)\,,
\end{equation}
\begin{equation}
   \xi \coloneqq \frac{1}{\sqrt{M_{23}}}(\sqrt{m_2}x_2 + \sqrt{m_3}x_3),
\end{equation}
where $M_{23} \coloneqq m_2+m_3$. Note that $q^2+\xi^2 = x_2^2+x_3^2$. Direct computations show that in $(q,\xi)$--coordinates the operator $h_{23}$ takes the form
\begin{equation}
    h_{23} = -\partial_q^2 - \partial_{\xi}^2 + V_{23}\left( \left( \frac{\abs{q}^2}{\mu_{23}} + (a_2-a_3)^2 \right)^{1/2} \right)
\end{equation}
where the reduced mass $\mu_{23}$ of particles $(23)$ is given by
\begin{equation}
    \mu_{23}  \coloneqq \frac{m_2m_3}{M_{23}} \, .
\end{equation}
The set $K_{23}(\gamma)$ in $(x_1, q,\xi)$--coordinates is given by
\begin{equation} \label{K_23_in_new_coordinates}
\begin{split}
        K_{23}(\gamma) &= \left\{ (x_1,q,\xi)\in \R^3: \abs{q} \leq \gamma \sqrt{\mu_{23}} \left( x_1^2 + q^2 + \xi^2 \right)^{1/2} \right\} \\
        &=\left\{ (x_1,q,\xi)\in \R^3: \abs{q} \leq \kappa_0 \left( x_1^2 + \xi^2 \right)^{1/2} \right\} 
\end{split}
\end{equation}
with
\begin{equation} \label{def_kappa_0}
    \kappa_0 \coloneqq \left(\frac{\gamma^2 \mu_{23}}{1-\gamma^2 \mu_{23}} \right)^{1/2} \, .
\end{equation}
The functional $L_{23}[\psi_{23}]$ can be written as
\begin{equation} \label{q_xi_introduced}
    \begin{split}
        L_{23}&[\psi_{23}]  \geq  \left( \frac{M_{23}}{\mu_{23}} \right)^{1/2} \Bigg[ \int_{K_{23}(\gamma)} \abs{\partial_{q} \psi_{23}}^2 +  V_{23} \abs{\psi_{23}}^2 d(x_1,q,\xi)\\
        &+\int_{K_{23}(\gamma)}\abs{\nabla_{(x_1,\xi)} \psi_{23}}^2 d(x_1,q,\xi) -  2 \int_{K_{23}(\gamma)\setminus S(0,b)} \frac{\abs{\psi_{23}}^{2}}{\abs{(x_1,q,\xi))}^{2+\tau}} d(x_1,q,\xi) \Bigg]
    \end{split}
\end{equation}
where $(M_{23}/\mu_{23})^{1/2}$ is the Jacobian determinant of the transformation to the new set of coordinates $(x_1,q,\xi)$. In abuse of notation, we denote $\psi_{23}$ expressed in coordinates $(x_1,q,\xi)$ by the same letter. We estimate \eqref{q_xi_introduced} in two steps. As the first step, we show that 
\begin{lemma} \label{first_lemma_in_32}
    \begin{equation} \label{first_step_in_32}
    \int_{K_{23}(\gamma)} \left( \abs{\partial_{q} \psi_{23}}^2 +  V_{23} \abs{\psi_{23}}^2 \right) d(x_1,q,\xi) \geq -C\int_{\partial K_{23}(\gamma)}  \frac{\abs{\psi_{23}}^2 }{\abs{x}^{1+\delta}} d\sigma 
\end{equation}
for some $C>0$ independent of $\psi$.
\end{lemma}
As a second step, we show
\begin{lemma} \label{second_step_in_32}
    \begin{equation} 
    \label{rmainder_n_estimated}
    \begin{split}
     \mathcal{N}[\psi_{23}] \coloneqq &\int_{K_{23}(\gamma)}  \abs{\nabla_{(x_1,\xi)} \psi_{23}}^2 d(x_1,q,\xi) -  2 \int_{K_{23}(\gamma)\setminus S(0,b)} \frac{\abs{\psi_{23}}^{2}}{\abs{(x_1,q,\xi))}^{2+\tau}} d(x_1,q,\xi) \\
     &\geq -C\int_{\partial K_{23}(\gamma)}  \frac{\abs{\psi_{23}}^2 }{\abs{x}^{1+\tau}} d\sigma   
    \end{split}
    \end{equation}
for some constant $C>0$ independent of $\psi$.
\end{lemma}
\subsubsection{Proof of Lemma~\ref{first_lemma_in_32}}
We rewrite the first integral on the right--hand side of \eqref{q_xi_introduced} as
\begin{equation} \label{splitted_integrals}
\int_{K_{23}(\gamma)} \left( \abs{\partial_{q} \psi_{23}}^2 +  V_{23} \abs{\psi_{23}}^2 \right) d(x_1,q,\xi) = \int_{\R^2} \int_{-\kappa_0 \abs{(x_1,\xi)}}^{\kappa_0 \abs{(x_1,\xi)}} \abs{\partial_{q} \psi_{23}}^2 +  V_{23} \abs{\psi_{23}}^2 dq 
 \, d(x_1,\xi) \, .
\end{equation}
Note that due to the positivity of operator $h_{23}$ also holds
\begin{equation}
     -\partial_q^2  + V_{23}\left( \left( \frac{\abs{q}^2}{\mu_{23}} + (a_2-a_3)^2 \right)^{1/2} \right) \geq 0 \, .
\end{equation}
Since the potential $V_{23}$ satisfies 
\begin{equation}
    \abs{V_{23}\left( \left( \frac{\abs{q}^2}{\mu_{23}} + (a_2-a_3)^2 \right)^{1/2} \right)} \leq \tilde C \abs{q}^{2+\delta} \, 
\end{equation}
for some constant $\tilde C>0$ if $q$ is sufficiently large, applying \cite[Lemma 6.2]{BBV:2022} (for convenience of the reader we proof it in Appendix \ref{app:D}) there exists $C>0$ that depends on $\delta$ but is independent of $\psi$ such that
\begin{equation} \label{used_pos_of_1d_op}
\begin{split}
   \int_{-\kappa_0 \abs{(x_1,\xi)}}^{\kappa_0 \abs{(x_1,\xi)}} &\abs{\partial_q \psi_{23}}^2 +V_{23} \abs{\psi_{23}}^2 dq  \\
   &\geq -C \frac{\abs{\psi_{23}(x_1, \kappa_0 \abs{(x_1,\xi)}, \xi)}^2  +\abs{\psi_{23}(x_1, -\kappa_0 \abs{(x_1,\xi)}, \xi)}^2  }{\abs{(x_1,\xi)}^{1+\delta}} \, . 
\end{split}
\end{equation}
Combining \eqref{splitted_integrals} and \eqref{used_pos_of_1d_op} we find
\begin{equation} \label{lemma_bbv_applied}
    \begin{split}
        \int_{K_{23}(\gamma)} &\left( \abs{\partial_{q} \psi_{23}}^2 +  V_{23} \abs{\psi_{23}}^2 \right) d(x_1,q,\xi) \\
        &\geq -C\int_{\R^2}  \frac{\abs{\psi_{23}(x_1, \kappa_0 \abs{(x_1,\xi)}, \xi)}^2  +\abs{\psi_{23}(x_1, -\kappa_0 \abs{(x_1,\xi)}, \xi)}^2  }{\abs{(x_1,\xi)}^{1+\delta}} d(x_1,\xi) \, .
    \end{split}
\end{equation}
The points $(x_1, \pm \kappa_0 \abs{(x_1,\xi)}, \xi)$ belong to the  surface $\partial K_{23}(\gamma)$. Direct computations show that for the surface measure $d\sigma$ associated with the set $\partial K_{23}(\gamma)$ satisfies the relation 
\begin{equation} \label{surface_measure_with_kappa}
    d\sigma =  \kappa_0  d(x_1,\xi),
\end{equation}
with $\kappa_0$ defined in \eqref{def_kappa_0}. Consequently from \eqref{lemma_bbv_applied} we find 
\begin{equation} \label{got_surface_3.2}
    \int_{K_{23}(\gamma)} \left( \abs{\partial_{q} \psi_{23}}^2 +  V_{23} \abs{\psi_{23}}^2 \right) d(x_1,q,\xi) \geq -C\int_{\partial K_{23}(\gamma)}  \frac{\abs{\psi_{23}}^2 }{\abs{(x_1,\xi)}^{1+\delta}} d\sigma 
\end{equation}
for a possible different constant $C>0$.
Using that on $\partial K_{23}(\gamma)$
\begin{equation}
    \abs{x}^2= \abs{(x_1,\xi)}^2 + q^2 = (1+\kappa_0^2) \abs{(x_1,\xi)}^2
\end{equation}
completes the proof of Lemma~\ref{first_lemma_in_32}. \qed 
\begin{figure}[!ht]
    \begin{tikzpicture}[scale=1] 
    \fill[orange,opacity=0.4] (0,0) -- (4,-2) -- (4,2) -- cycle;
    
    \draw[black] (0,0) circle (1.5cm);
    \draw[fill, blue!20,opacity=0.5] (0,0) circle (1.5cm);
    \draw [thick] (0,0) -- (4,2);
    \draw [thick] (0,0) -- (4,-2);


    \draw[fill, white] (0,0) circle (1.5cm);
    \draw[black] (0,0) circle (1.5cm);
    \draw[fill, blue!20,opacity=0.5] (0,0) circle (1.5cm);

    \draw[-] (0.75,0) -- (0.75,0.1);
    \draw[-] (0.75,0) -- (0.75,-0.1);

    \draw[red,-] (2,1) -- (4.5,1);

    \node at (3.2,-0.5) {$K_{23}(\gamma) $};

    \node at (1.1,0,0) [below] {$b/2$};
        
    \draw [ ->] (0,0) -- (0,2.2);
    \node at (0,2.4) {$q$};
    \draw[->] (0,0) -- (4.5,0) node[anchor=west] {$\rho$};

    \draw [dotted] (0,0) -- (4,2) node[pos=0.7,above, sloped] {$q=\kappa_0 \rho$};
    \draw [dotted] (0,0) -- (4,-2);

    \draw [<->] (0,0) -- (-1.47,0.3);
    \node at (-0.75,-0.05) {$b$};

\draw [thick] (-0.5,2,0) arc[start angle=180, end angle=360, x radius=0.5, y radius=0.125];
\draw [->,thick] (0.5,2,0) arc[start angle=0, end angle=60, x radius=0.5, y radius=0.125];
\end{tikzpicture}
\caption{Sketch of the set $K_{23}(\gamma)$. In the circular blue area, the function $\psi$ vanishes. For fixed $q\in \R$ the horizontal red line indicates the path of integration used in Lemma~\ref{second_step_in_32}.}\label{fig:003}
\end{figure}
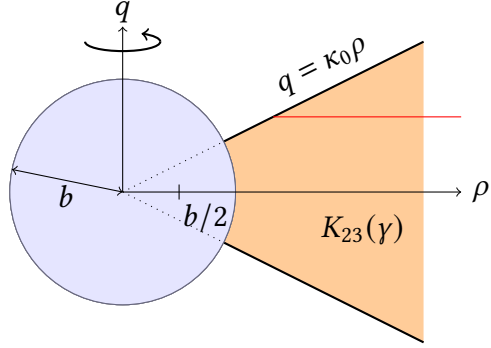
\subsubsection{Proof of Lemma~\ref{second_step_in_32}} 
\label{lemma_with_polars}
This lemma mainly follows the ideas of \cite[Lemma 6.7]{BBV:2022}. We introduce polar coordinates in the $(x_1,\xi)$--plane as
\begin{equation}
    \rho \coloneqq \sqrt{x_1^2+\xi^2}, \quad  \varphi \coloneqq \arctan(x_1/\xi) \, .
\end{equation}
The set $\partial K_{23}(\gamma)$ corresponds to points with 
\begin{equation}
    \abs{q} = \kappa_0 \rho,
\end{equation}
where $\kappa_0$ was defined in \eqref{def_kappa_0}. For each fixed $q\in \R$ let
\begin{equation}
    \overline{\psi}(q) \coloneqq 
        \int_0^{2\pi } \psi_{23}(q, \kappa_0^{-1}\abs{q},\varphi) \frac{d\varphi}{2\pi} \quad \text{ and } \quad \overline{\psi}_1(q,\rho, \varphi) \coloneq \overline{\psi}(q) \cdot \mathbbm{1}(\rho,\varphi) \, .
\end{equation}
Let $\mathscr{F} \coloneqq \psi_{23}-\overline{\psi}_1$.
 We write $\nabla_{(\rho,\varphi)}$ for the gradient in polar coordinates in the $(x_1,\xi)$--plane. Then $\nabla_{(\rho,\varphi)} \overline{\psi}_1 \equiv 0$ and consequently
\begin{equation} \label{symm_ovpsi_psi_}
    \begin{split}
        \int_{K_{23}(\gamma)}  &\abs{\nabla_{(x_1,\xi)} \psi_{23}}^2 d(x_1,q,\xi) \\
        &=\int_{K_{23}(\gamma)}  \abs{\nabla_{(\rho,\varphi)} \psi_{23}}^2 d(q,\rho,\varphi)  = \int_{K_{23}(\gamma)}  \abs{\nabla_{(\rho,\varphi)} \mathscr{F} }^2d(q,\rho,\varphi) ,
    \end{split}
\end{equation}
where $d(x_1,q,\xi)=d(q,\rho,\varphi) =\rho \, dq d\rho d\varphi$. Inserting \eqref{symm_ovpsi_psi_} into the right--hand side of \eqref{rmainder_n_estimated} we arrive at
\begin{equation}\label{N_est_second_step}
    \mathcal{N}[\psi_{23}] =  \int_{K_{23}(\gamma)}  \abs{\nabla_{(\rho,\varphi)} \mathscr{F} }^2d(q,\rho,\varphi) -  2 \int_{K_{23}(\gamma)\setminus S(0,b)} \frac{\abs{\psi_{23}}^{2}}{\abs{(x_1,q,\xi)}^{2+\tau}} d(x_1,q,\xi) \, .
\end{equation}
Transforming the second integral on the right--hand side of \eqref{N_est_second_step} to polar coordinates we get
\begin{equation}\label{N_est_second_step_polar}
    \mathcal{N}[\psi_{23}] =  \int_{K_{23}(\gamma)}  \abs{\nabla_{(\rho,\varphi)} \mathscr{F} }^2d(q,\rho,\varphi) -  2 \int_{K_{23}(\gamma)\setminus S(0,b)} \frac{\abs{\psi_{23}}^{2}}{\abs{(q,\rho)}^{2+\tau}} d(q,\rho,\varphi) \, .
\end{equation}
Since  $\mathscr{F} = \psi_{23}-\overline{\psi}_1$
and $(a+b)^2 \leq 2a^2+2b^2$ it holds
\begin{equation} \label{quadratic_ovpsi_psi1}
    \begin{split}
         &\int_{K_{23}(\gamma)\setminus S(0,b)} \frac{\abs{\psi_{23}}^{2}}{\abs{(q,\rho)}^{2+\tau}} d(q,\rho,\varphi) \\
         &\leq 2 \int_{K_{23}(\gamma)\setminus S(0,b)} \frac{\abs{\mathscr{F} }^{2}}{\abs{(q,\rho)}^{2+\tau}} d(q,\rho,\varphi) +  2 \int_{K_{23}(\gamma)\setminus S(0,b)} \frac{\abs{\overline{\psi}_1}^{2}}{\abs{(q,\rho)}^{2+\tau}} d(q,\rho,\varphi) \, .
    \end{split}
\end{equation}
Combining \eqref{N_est_second_step_polar} and \eqref{quadratic_ovpsi_psi1} yields
\begin{equation} \label{N_third_step}
    \begin{split}
        \mathcal{N}[\psi_{23}] \geq &\int_{K_{23}(\gamma)}  \abs{\nabla_{(\rho,\varphi)} \mathscr{F} }^2 d(q,\rho,\varphi)-4 \int_{K_{23}(\gamma)\setminus S(0,b)} \frac{\abs{\mathscr{F} }^{2}}{\abs{(q,\rho)}^{2+\tau}} d(q,\rho,\varphi) \\
        &-  4 \int_{K_{23}(\gamma)\setminus S(0,b)} \frac{\abs{\overline{\psi}_1}^{2}}{\abs{(q,\rho)}^{2+\tau}} d(q,\rho,\varphi) \, .
    \end{split}
\end{equation}
Next, we show that the sum of the first two integrals on the right--hand side of \eqref{N_third_step} is positive. The function $\psi_{23}=0$ for $\abs{x}<b$ and consequently for 
\begin{equation}
    \rho_0 \coloneqq \max\{\kappa_0^{-1} \abs{q},b/2\}
\end{equation}
we have
\begin{equation} \label{this_2d_kin_energy}
    \begin{split}
        \int_{K_{23}(\gamma)}  \abs{\nabla_{(\rho,\varphi)} \mathscr{F} }^2 d(q,\rho,\varphi) &= \int_{-\infty}^\infty \int_{\rho_0}^\infty \int_0^{2\pi} \abs{\nabla_{(\rho,\varphi)} \mathscr{F} (q,\rho,\varphi)}^2 d(\rho,\varphi) dq \, .
    \end{split}
\end{equation}
For \(\rho = \rho_0\), the projection of \(\mathscr{F}\) onto functions with zero angular momentum in the \((x_1, \xi)\)--plane vanishes. As a result, the following two--dimensional Hardy inequality (see Lemma~\ref{2d_2_hardy_inq}) holds for almost all \(q \in \mathbb{R}\):
\begin{equation} \label{this_2d_hardy}
    \begin{split}
        \int_{\rho_0}^\infty \int_0^{2\pi} \left|\nabla_{(\rho,\varphi)} \mathscr{F} (q, \rho, \varphi)\right|^2 \, d(\rho, \varphi) \geq \frac{1}{4} \int_{\rho_0}^\infty \int_0^{2\pi} \frac{\left|\mathscr{F} (q, \rho, \varphi)\right|^2}{\rho^2\left(1 + \ln^2(\rho)\right)} \, d(\rho, \varphi),
    \end{split}
\end{equation}
where we assume \(\rho_0 > 1\) for sufficiently large \(b > 0\). Inserting \eqref{this_2d_hardy} into the right--hand side of \eqref{this_2d_kin_energy} gives
\begin{equation} \label{use_for_N_fourth}
     \int_{K_{23}(\gamma)}  \abs{\nabla_{(\rho,\varphi)} \mathscr{F} }^2 d(q,\rho,\varphi) \geq \frac{1}{4} \int_{K_{23}(\gamma)} \frac{\abs{\mathscr{F}}^2}{\rho^2\left(1+\ln^2\left(\rho \right)\right)}   d(q,\rho,\varphi) \, .
\end{equation} 
Since $\rho>b/2$ on $K_{23}(\gamma)\cap \supp{\psi_{23}}$ the positivity of the sum of the first two terms on the right--hand side of \eqref{N_third_step} follows from \eqref{use_for_N_fourth} for $b>0$ sufficiently large. We arrive at
\begin{equation} \label{N_fourth_step}
    \begin{split}
        &\mathcal{N}[\psi_{23}] \geq -  4 \int_{K_{23}(\gamma)\setminus S(0,b)} \frac{\abs{\overline{\psi}_1}^{2}}{\abs{(q,\rho)}^{2+\tau}} d(q,\rho,\varphi) \, .
    \end{split}
\end{equation}
It remains to show that the integral on the right--hand side of \eqref{N_fourth_step} can be estimated by an integral over $\partial K_{23}(\gamma)$. By direct computation
\begin{equation} \label{N_step_got_integrals}
    \int_{K_{23}(\gamma)\setminus S(0,b)} \frac{\abs{\overline{\psi}_1(q,\rho,\varphi)}^{2}}{\abs{(q,\rho)}^{2+\tau}} d(q,\rho,\varphi) = 2\pi \int_{-\infty}^\infty \int_{\rho_0}^{\infty } \frac{\abs{\overline{\psi}_1(q,\rho,\varphi)}^{2}}{\abs{(q,\rho)}^{2+\tau}} \rho d\rho dq \, .
\end{equation}
Using the definition of $\overline{\psi}_1$ and Schwarz Inequality
\begin{equation} \label{N_step_almost_surface}
     2\pi \int_{-\infty}^\infty \int_{\rho_0}^{\infty } \frac{\abs{\overline{\psi}_1(q)}^{2}}{\abs{(q,\rho)}^{2+\tau}} \rho d\rho dq \leq  \int_{-\infty}^\infty \int_{\rho_0}^{\infty }\frac{\int_{0}^{2\pi} \abs{\psi_{23}(q, \kappa_0^{-1}\abs{q},\varphi)}^{2} d\varphi }{\abs{(q,\rho)}^{2+\tau}} \rho d\rho dq \, . 
\end{equation}
Combining \eqref{N_step_got_integrals} and \eqref{N_step_almost_surface} we arrive at
\begin{equation}
    \begin{split}
        \int_{K_{23}(\gamma)\setminus S(0,b)}& \frac{\abs{\overline{\psi}_1}^{2}}{\abs{(q,\rho)}^{2+\tau}} d(q,\rho,\varphi) \\
        &\leq \int_{-\infty}^\infty  \int_{0}^{2\pi} \abs{\psi_{23}(q, \kappa_0^{-1}\abs{q},\varphi)}^{2} d\varphi  \int_{\rho_0}^{\infty }\frac{1}{\abs{(q,\rho)}^{2+\tau}} \rho d\rho dq 
    \end{split}
\end{equation}
Using $| (q, \rho) | > \rho$ and $\rho_0 > \kappa_0^{-1} |q|$, yields by solving the integral over $\rho$ that there exists a constant $C > 0$, which depends on $\tau$ and $\kappa_0$ but is independent of $\psi$, such that
\begin{equation} \label{surface_int_for_N}
    \int_{K_{23}(\gamma) \setminus S(0,b)} \frac{|\overline{\psi}_1|^2}{|(q,\rho)|^{2+\tau}} \, d(q,\rho,\varphi) \leq C \int_{-\infty}^\infty \int_0^{2\pi} \frac{|\psi_{23}(q, \kappa_0^{-1}|q|,\varphi)|^2}{|q|^{1+\tau}} \, q \, dq \, d\varphi.
\end{equation}
The points $(q, \kappa_0^{-1} |q|, \varphi)$ with $q \in \mathbb{R}$ and $\varphi \in [0, 2\pi)$ correspond to points on the surface $\partial K_{23}(\gamma)$. Consequently, by substituting \eqref{surface_int_for_N} into \eqref{N_fourth_step}, we obtain for another constant $C > 0$ that depends on $\kappa_0$ and $\tau$ but is independent of $\psi$:
\begin{equation}\label{N_done}
    \begin{split}
     \mathcal{N}[\psi_{23}] \geq -C \int_{\partial K_{23}(\gamma)} \frac{|\psi|^2}{|x|^{1+\tau}} \, d\sigma,  
    \end{split}
\end{equation}
where we used that $\psi_{23}=\psi$ on $\partial K_{23}(\gamma)$. This completes the proof of Lemma~\ref{second_step_in_32} and as discussed in \ref{proof_of_L_23_lemma} this also completes the proof of Lemma~\ref{L_23_lemma}. \qed
\subsection{Proof of Lemma~\ref{trace_lemma_001}} \label{proof_of_lemma_trace_001}
We prove Lemma~\ref{trace_lemma_001} for $\alpha=(12)$. The proof for $\alpha=(13)$ is similar. We introduce spherical coordinates as follows:
\begin{equation} \label{sphericals_for_12}
    \begin{split}
        x_1 &= \abs{x} \sin(\varphi) \sin(\theta), \\
        x_2 &= \abs{x} \cos(\varphi) \sin(\theta), \\
        x_3 &= \abs{x} \cos(\theta),
    \end{split}
\end{equation}
where $\varphi \in [0, 2\pi)$, $\theta \in [0, \pi]$, and $\abs{x} \in [0, \infty)$. The boundary $\partial K_{12}(\gamma)$ is then described by $(\abs{x}, \pm \theta_0, \varphi)$, where $\theta_0 = \arcsin(\gamma)$ and $\theta_0 \in [0, \pi]$ (see Figure~\ref{fig:002}).

Let $P_0[(12)]$ be the projection onto radially symmetric functions in the $(x_1,x_2)$--plane introduced in Lemma~\ref{L_1j_lemma}, then
\begin{equation} \label{get_rid_of_projection}
   \int_{\partial K_{12}(\gamma)} \frac{\abs{(P_0[(12)] \psi)(x)}^2}{\abs{x}^{1+\delta}} \, d\sigma \leq \int_{\partial K_{12}(\gamma)} \frac{\abs{\psi(x)}^2}{\abs{x}^{1+\delta}} \, d\sigma 
\end{equation}
where the surface measure $d\sigma$ associated with $\partial K_{12}(\gamma)$ in the spherical coordinattes \eqref{sphericals_for_12} is given by 
\begin{equation} \label{surface_measure}
    d\sigma = \abs{\sin{\theta_0}} \abs{x} d\abs{x} d\varphi
\end{equation}
By applying the one--dimensional trace theorem in the $\theta$ variable (see, \cite[Theorem 1, p. 272]{book:Evans}) for every fixed $\abs{x} > b/2$, $\varphi\in[0,2\pi)$, there is a constant $C>0$ that depends on $\theta_0$ and $\theta_1\coloneqq \arcsin{(\gamma_1)}$  but is independent of $\abs{x}$ and $\varphi$ such that
\begin{equation} \label{trace_thrm}
    \frac{\abs{\psi(\abs{x}, \theta_0, \varphi)}^2}{\abs{x}^{1+\delta}} \leq C \int_{\theta_0}^{\theta_1} \frac{ \abs{\psi(\abs{x}, \theta, \varphi)}^2 + \abs{(\partial_\theta \psi)(\abs{x}, \theta, \varphi)}^2}{\abs{x}^{1+\delta}}   d\theta.
\end{equation}
Integrating \eqref{trace_thrm} with respect to $d\sigma$ in \eqref{surface_measure} yields
\begin{equation} \label{pos_theta_0}
    \begin{split}
        &\int_{0}^{2\pi} \int_{b/2}^\infty \abs{\psi(\abs{x}, \theta_0, \varphi)}^2 \abs{\sin{\theta_0}} \abs{x} d\abs{x} d\varphi \\
        &\leq C \abs{\sin{\theta_0}} \int_{0}^{2\pi} \int_{\theta_0}^{\theta_1} \int_{b/2}^\infty  \frac{ \abs{\psi(\abs{x}, \theta, \varphi)}^2 + \abs{(\partial_\theta \psi)(\abs{x}, \theta, \varphi)}^2}{\abs{x}^{1+\delta}}   \abs{x} d\abs{x} d\theta d\varphi  .
    \end{split}
\end{equation}
where we have also used that $\psi(\abs{x}, \theta, \varphi) = 0$ for $\abs{x} < b/2$. Similar we get
\begin{equation} \label{neg_theta_0}
    \begin{split}
        &\int_{0}^{2\pi} \int_{b/2}^\infty \abs{\psi(\abs{x}, -\theta_0, \varphi)}^2 \abs{\cos{\theta}} \abs{x} d\abs{x} d\varphi \\
        &\leq C \abs{\sin{\theta_0}} \int_{0}^{2\pi} \int_{-\theta_1}^{-\theta_0} \int_{b/2}^\infty  \frac{ \abs{\psi(\abs{x}, \theta, \varphi)}^2 + \abs{(\partial_\theta \psi)(\abs{x}, \theta, \varphi)}^2}{\abs{x}^{1+\delta}}   \abs{x} d\abs{x} d\theta d\varphi  .
    \end{split}
\end{equation}
For $K_{12}(\gamma,\gamma_1) = K_{12}(\gamma_1) \setminus K_{12}(\gamma)$ by combining \eqref{pos_theta_0} and \eqref{neg_theta_0} we find
\begin{equation} \label{integrals_over_anulus_before}
    \begin{split}
        &\int_{\partial K_{12}(\gamma)} \frac{\abs{\psi(x)}^2}{\abs{x}^{1+\delta}} \, d\sigma  \\
        &\leq C \abs{\sin{\theta_0}} \int_{ K_{12}(\gamma, \gamma_1)}  \frac{ \abs{\psi(\abs{x}, \theta, \varphi)}^2 + \abs{(\partial_\theta \psi)(\abs{x}, \theta, \varphi)}^2}{\abs{x}^{1+\delta}}   \abs{x} d\abs{x} d\theta d\varphi \, .
    \end{split}
\end{equation} 
Using $dx= \sin(\theta) \abs{x}^2 d\abs{x} d\theta d\varphi$ and $\sin(\theta_0)\leq \abs{\sin(\theta)}$ on $K_{12}(\gamma,\gamma_1)$ we arrive at
\begin{equation} \label{integrals_over_anulus}
    \begin{split}
        &\int_{\partial K_{12}(\gamma)} \frac{\abs{\psi(x)}^2}{\abs{x}^{1+\delta}} \, d\sigma  \\
        &\leq C \int_{ K_{12}(\gamma, \gamma_1)} \frac{ \abs{\psi(\abs{x}, \theta, \varphi)}^2 }{\abs{x}^{2+\delta}}  dx + C \int_{ K_{12}(\gamma, \gamma_1)} \frac{ \abs{(\partial_\theta \psi)(\abs{x}, \theta, \varphi)}^2 }{\abs{x}^{2+\delta}}  dx \, .
    \end{split}
\end{equation} 
Since $\psi(x) = 0$ for $\abs{x}\leq b/2$ we conclude from \eqref{integrals_over_anulus}
\begin{equation} \label{integrals_over_anulus_2}
    \begin{split}
        &\int_{\partial K_{12}(\gamma)} \frac{\abs{\psi(x)}^2}{\abs{x}^{1+\delta}} \, d\sigma  \\
        &\leq \frac{C}{(b/2)^{\delta}} \left( \int_{ K_{12}(\gamma, \gamma_1)} \frac{ \abs{\psi(\abs{x}, \theta, \varphi)}^2 }{\abs{x}^{2}}  dx + \int_{ K_{12}(\gamma, \gamma_1)} \frac{ \abs{(\partial_\theta \psi)(\abs{x}, \theta, \varphi)}^2 }{\abs{x}^{2}}  dx \right)
    \end{split}
\end{equation} 
Applying the radial Hardy Inequality to the first term on the right--hand side of \eqref{integrals_over_anulus_2} yields
\begin{equation} \label{integrals_over_anulus_3}
    \begin{split}
        &\int_{\partial K_{12}(\gamma)} \frac{\abs{\psi(x)}^2}{\abs{x}^{1+\delta}} \, d\sigma  \\
        &\leq \frac{C}{(b/2)^{\delta}} \left( \int_{ K_{12}(\gamma, \gamma_1)}\frac{1}{4} \abs{\partial_{\abs{x}}\psi(\abs{x}, \theta, \varphi)}^2   dx + \int_{ K_{12}(\gamma, \gamma_1)} \frac{ \abs{(\partial_\theta \psi)(\abs{x}, \theta, \varphi)}^2 }{\abs{x}^{2}}  dx \right)
    \end{split}
\end{equation} 
Expressing the gradient in spherical coordinates we conclude from \eqref{integrals_over_anulus_3} that
\begin{equation} \label{integrals_over_anulus_4}
    \begin{split}
        \int_{\partial K_{12}(\gamma)} \frac{\abs{\psi(x)}^2}{\abs{x}^{1+\delta}} \, d\sigma \leq \frac{C}{(b/2)^{\delta}} \norm{ \nabla \psi }_{L^2(K_{12}(\gamma,\gamma_1))}^2,
    \end{split}
\end{equation} 
which completes the proof of Lemma~\ref{trace_lemma_001}  \qed
\subsection{Proof of Lemma~\ref{trace_lemma_002}} 
In the set of coordinates $(x_1,q,\xi)$ introduced in 
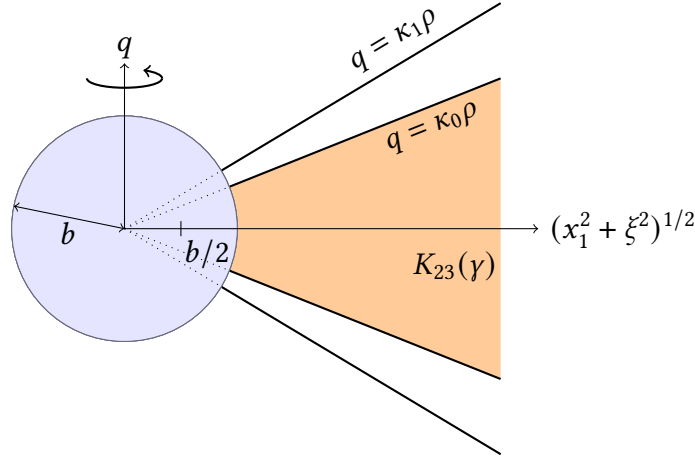
\begin{figure}[!ht]
    \begin{tikzpicture}[scale=1] 
    \fill[orange,opacity=0.4] (0,0) -- (5,-2) -- (5,2) -- cycle;
    
    \draw[black] (0,0) circle (1.5cm);
    \draw[fill, blue!20,opacity=0.5] (0,0) circle (1.5cm);
    \draw [thick] (0,0) -- (5,2);
    \draw [thick] (0,0) -- (5,-2);

    \draw [thick] (0,0) -- (5,3);
    \draw [thick] (0,0) -- (5,-3);

    \draw[fill, white] (0,0) circle (1.5cm);
    \draw[black] (0,0) circle (1.5cm);
    \draw[fill, blue!20,opacity=0.5] (0,0) circle (1.5cm);

    \draw[-] (0.75,0) -- (0.75,0.1);
    \draw[-] (0.75,0) -- (0.75,-0.1);
    
    \node at (4.4,-0.5) {$K_{23}(\gamma) $};

    \node at (1.1,0,0) [below] {$b/2$};
        
    \draw [ ->] (0,0) -- (0,2.2);
    \node at (0,2.4) {$q$};
    \draw[->] (0,0) -- (5.5,0) node[anchor=west] {$(x_1^2+\xi^2)^{1/2}$};

    \draw [dotted] (0,0) -- (5,2) node[pos=0.8,below, sloped] {$q=\kappa_0 \rho$};
    \draw [dotted] (0,0) -- (5,-2);

    \draw [dotted] (0,0) -- (5,3) node[pos=0.75,above, sloped] {$q=\kappa_1 \rho$};
    \draw [dotted] (0,0) -- (5,-3);
    
    \draw [<->] (0,0) -- (-1.47,0.3);
    \node at (-0.75,-0.05) {$b$};

\draw [thick] (-0.5,2,0) arc[start angle=180, end angle=360, x radius=0.5, y radius=0.125];
\draw [->,thick] (0.5,2,0) arc[start angle=0, end angle=60, x radius=0.5, y radius=0.125];
\end{tikzpicture}
\caption{Sketch of the sets $K_{23}(\gamma)$ and $K_{23}(\gamma_1)$ and their relation to the constants $\kappa_0$ and $\kappa_1$. This construction is used in Lemma~\ref{trace_lemma_002}.}\label{fig:004}
\end{figure}\ref{proof_of_L_23_lemma} the set $\partial K_{23}(\gamma)$ is determined by
\begin{equation}
    \abs{q} = \kappa_0(x_1^2+\xi^2)^{1/2}
\end{equation}
with $\kappa_0$ defined in \eqref{def_kappa_0}. We introduce spherical coordinates as follows:
\begin{equation} \label{sphericals_for_12_vartheta}
    \begin{split}
        x_1 &= \abs{x} \sin(\varphi) \cos(\vartheta), \\
        \xi &= \abs{x} \cos(\varphi) \cos(\vartheta), \\
        q &= \abs{x} \sin(\vartheta),
    \end{split}
\end{equation}
where we used that the coordinates $(x_1,q,\xi)$ fulfill $\abs{x} = \abs{(x_1,q,\xi)}$ (see Lemma~\ref{K_23_in_new_coos}). In this set of coordinates $\vartheta_0 \coloneqq \arctan(\kappa_0)$ corresponds to the opening angle of the conical set $K_{23}(\gamma)$. Let $\vartheta_1 \coloneqq \arctan{\kappa_1} $  with 
\begin{equation} \label{def_kappa_1}
    \kappa_1 \coloneqq \left(\frac{\gamma_1^2 \mu_{23}}{1-\gamma_1^2 \mu_{23}} \right)^{1/2},
\end{equation}
then $\vartheta_1>\vartheta_0$ corresponds to the opening angle of the conical set $K_{23}(\gamma_1)$ (see Figure \ref{fig:004}).
Analogous to Lemma~\ref{trace_lemma_001} by application of the one--dimensional trace theorem in the variable $\vartheta$ and the radial Hardy Inequality proves Lemma~\ref{trace_lemma_002}. \qed

\appendix
\section{Hardy--Type Inequalities}\label{app:A}
In the work at hand, several types of Hardy Inequalities are used. We collect them here from different sources and refer to their proofs.

The one--dimensional Hardy Inequality, for which a proof can be found in \mbox{\cite[Theorem 2.65]{book:FLW:2022}} is:
\begin{lemma} \label{one_dimensional_hardy}
    Let $u\in \dot H^1(\R)$ and assume $\displaystyle{\liminf_{r \to 0}}\abs{u(r)}=0$. Then
    \begin{equation} \label{one_d_hardy}
        \int_{\R} \frac{\abs{u(r)}^2}{\abs{r}^2} \, dr \leq 4 \int_{\R} \abs{u'(r)}^2 \, dr \, .
    \end{equation}
\end{lemma}
The following is the Hardy Inequality in dimension two for which a proof can be found in \cite{S:1994} or derived from \cite[Pop. 2.68]{book:FLW:2022}
\begin{lemma} \label{2d_2_hardy_inq} Let $u\in \dot H^1(\R^2)$ with 
\begin{equation}
    \int_{0}^{2\pi} u(c,\varphi) d\varphi=0 ,
\end{equation}
where $(\rho,\varphi)$ are polar coordinates and $c>0$. Then the following Hardy Inequality is true:
\begin{equation}
    \int_{\R^2} \frac{\abs{u}^2}{\abs{x}^2\left(1+\ln^2\left(\frac{\abs{x}}{c}\right) \right)} dx \leq 4 \int_{\R^2} \abs{\nabla u}^2 dx \, .
\end{equation}
\end{lemma}

The following is the radial Hardy Inequality on conical sets which was proven in \cite[§4]{N:2006}.
\begin{lemma} \label{radial_hardy_inequality}
    Let $d\geq 2$ and $K \subset \mathbb{S}^{d-1}$ a smooth domain and
    \begin{equation}
        \Omega \coloneqq  \{ x = \abs{x} \omega: \omega \in K  \} \,. 
    \end{equation}
    Then for any $u \in \dot H^1(\Omega\setminus\{0\})$ the following Hardy Inequality holds:
    \begin{equation}
        \int_{\Omega} \abs{u}^2 \abs{x}^{-2} dx \leq \left( \frac{2}{d-2} \right)^2 \int_\Omega \abs{\nabla u}^2 dx \, .
    \end{equation}
\end{lemma}
\section{Remarks on the Geometry of the Sets $K_{\alpha}(\gamma)$}\label{app:B}
\begin{lemma} \label{cones_are_disjoint} Let $\alpha,\beta \in \{(12),(13),(23)\}$ then
    \begin{equation} \label{app:eq:b1}
       K_{\alpha}(\gamma) \cap K_{\beta}(\gamma) = \{0\}
    \end{equation}
    for $\alpha \neq \beta$ if $\gamma>0$ is small enough.
\end{lemma}
\begin{proof} We begin with the sets $K_{1j}(\gamma)$ for $j\in \{2,3\}$. Assume that $x\in K_{12}(\gamma) \cap K_{13}(\gamma)$ and $x\neq 0$ then
\begin{equation}
     \abs{x}^2 \leq (x_1^2 + x_2^2) + (x_1^2+x_3^3) < 2\gamma^2 \abs{x}^2
\end{equation}
which fails for $\gamma$ small enough, which we assume henceforth. Consequently, we have 
\begin{equation}
    K_{12}(\gamma) \cap K_{13}(\gamma)= \{0\}
\end{equation}
Next we assume that $x\in K_{12}(\gamma) \cap K_{23}(\gamma)$ and $x\neq 0$, then
\begin{equation}
    \begin{split}
        x_1^2 + x_2^2 &< \gamma^2 \abs{x}^2 \, , \\
        \abs{\frac{x_2}{\sqrt{m_2}} - \frac{x_3}{\sqrt{m_3}}} &< \gamma \abs{x}  \, .
    \end{split}
\end{equation}
Let $m \coloneqq \min\{m_1,m_2,m_3\}$ and $M \coloneqq \max \{m_1,m_2,m_3\}$ and recall the choice $ \sqrt{m_i}y_i =x_i$ then 
\begin{equation} \label{cones_in_y}
    \begin{split}
        y_1^2 + y_2^2 &< \frac{\gamma^2}{m} M \abs{y}^2 \, \\
        \abs{y_2 - y_3} &< \gamma \sqrt{M} \abs{y}  \, .
    \end{split}
\end{equation}
Note that the points $(y_2,0), (y_3,0)$ and $(0,y_1)$ are the corners of a triangle in $\R^2$ and thus by the triangle inequality
\begin{equation} \label{triangle_inequality}
    \sqrt{y_1^2+ y_3^2} \leq \sqrt{y_1^2 + y_2^2} +\abs{y_2 - y_3} < \gamma \sqrt{M}\left( 1 + m^{-1/2} \right) \abs{y} \, .
\end{equation}
Combining \eqref{cones_in_y} and \eqref{triangle_inequality} there exists a constant $c(M,m)>0$ independent of $\gamma$ and $y\neq 0$ such that
\begin{equation}
    \abs{y}^2 \leq (y_1^2 +y_2^2)+ (y_1^2+y_3^2) \leq c(M,m) \abs{\gamma}^2 \abs{y}^2
\end{equation}
which fails for $\gamma$ small enough  and thus $K_{12}(\gamma) \cap K_{23}(\gamma) = \{0\}$. Repeating the same arguments for the set $K_{13}(\gamma)\cap K_{23}(\gamma)$ concludes the proof.
\end{proof}

\begin{lemma} \label{potentials_small_outside_cones}
    Let $\alpha \in \{(12),(13),(23)\}$ and let $\gamma>0$ satisfy the condition \eqref{app:eq:b1}. Assume that the potentials $V_\alpha$ satisfy \eqref{short-range}.
    Then there exists $C>0$ such that for any $x\notin K_\alpha(\gamma)$ with $\abs{x}>b$ holds 
    \begin{equation} \label{short_range_with_x}
        \abs{V_\alpha(x)} \leq C \abs{x}^{-\nu_\alpha}, 
    \end{equation}
    where $\nu_{23} = 2+ \delta$ and $ \nu_{12}=\nu_{13} =3+\delta$ .
\end{lemma}
\begin{remark} Lemma~\ref{potentials_small_outside_cones} implies that for small $\gamma > 0$ condition \eqref{short_range_with_x} holds on $\Omega(\gamma)$ for any of the potentials if $\abs{x}>b$ and $b>0$ large enough. Consequently the condition \eqref{short_range_with_x} applies on any of the sets $K_{\alpha}(\gamma,\tilde \gamma)$ defined in Section \ref{proof_of_main_thrm}. 
\end{remark}
\begin{proof}
    We begin with $\alpha = (23)$. The potential $V_{23}$ fulfills \eqref{short-range} and thus there exist constants $C,\delta>0$ and $A>0$ such that
    \begin{equation} \label{v_23_short_range}
    \left| V_{23}(\abs{\mathbf{r}_{23}}) \right| \leq C (1 + \abs{\mathbf{r}_{23}})^{-2-\delta},
\end{equation}
whenever \( \abs{\mathbf{r}_{23}} > A \). For $x \notin K_{23}(\gamma)$ we have
\begin{equation} \label{dist_23_est}
    \abs{\mathbf{r}_{23}} \geq \abs{\frac{x_2}{\sqrt{m_2}} - \frac{x_3}{\sqrt{m_3}} } \geq \gamma \abs{x} \, .
\end{equation}
Consequently $\abs{\mathbf{r}_{23}}>A$ for $\abs{x}>b$ for $b>0$ large enough. Combining \eqref{dist_23_est} and \eqref{v_23_short_range} the statement for $\alpha=(23)$ follows directly.

Next, we prove the case $\alpha=(12)$. The proof for $\alpha = (13)$ is similar. The potential $V_{12}$ fulfills \eqref{short-range} and thus there exist constants $C,\delta>0$ and $A>0$ such that
    \begin{equation} \label{v_12_short_range}
    \left| V_{12}(\abs{\mathbf{r}_{12}}) \right| \leq C (1 + \abs{\mathbf{r}_{12}})^{-3-\delta} 
\end{equation}
whenever \( \abs{\mathbf{r}_{12}} > A \). Assume $x\notin K_{12}(\gamma)$ then
\begin{equation} \label{lower_bnd_r_12}
    \begin{split}
        \abs{\mathbf{r}_{12}} &\geq \left( \frac{x_1^2}{m_1} + \frac{x_2^2}{m_2} - \cos(\zeta) 2\frac{x_1 x_2}{\sqrt{m_1 m_2}} \right)^{1/2} \\
        &\geq \left( \frac{x_1^2}{m_1} + \frac{x_2^2}{m_2}\right)^{1/2}(1-\cos(\zeta))^{1/2} \\
        &\geq \left(\min\left\{\frac{1}{m_1},\frac{1}{m_2} \right\}(1-\cos(\zeta))\right)^{1/2} \abs{(x_1,x_2)} \\
        &\geq \left(\min\left\{\frac{1}{m_1},\frac{1}{m_2} \right\}(1-\cos(\zeta)) \gamma \right)^{1/2} \abs{x}\, .
    \end{split}
\end{equation}
where we have used that $2ab\leq a^2+b^2$ from the first to the second line in \eqref{lower_bnd_r_12}. Since $\zeta \in (0,\pi/2]$, combining \eqref{lower_bnd_r_12} and \eqref{v_12_short_range} completes the proof.
\end{proof}
\begin{lemma} \label{K_23_in_new_coos}
    Let $(x_1,x_2,x_3) \in \R^3$ and let $(q,\xi)$ be defined as,
    \begin{equation}
    q = \frac{1}{\sqrt{M_{23}}}(\sqrt{m_3}x_2 - \sqrt{m_2}x_3)\,,
\end{equation}
\begin{equation}
   \xi = \frac{1}{\sqrt{M_{23}}}(\sqrt{m_2}x_2 + \sqrt{m_3}x_3),
\end{equation}
where $M_{23} \coloneqq m_2+m_3$.
Then
    \begin{equation} \label{invariant_measure}
        \abs{x} = \abs{(x_1,q,\xi)}\, .
    \end{equation}
    The surface measure on the set $\partial K_{23}(\gamma)$ in this set of coordinates is given by
    \begin{equation} \label{surf_element_23_apppendix}
        d\sigma = \kappa_0  d(x_1,\xi)
    \end{equation}
    with
    \begin{equation} \label{kappa_0_mu_23}
         \kappa_0 =\left(\frac{\gamma^2 \mu_{23}}{1-\gamma^2 \mu_{23}} \right)^{1/2} , \quad  \mu_{23} = \frac{m_2m_3}{m_2+m_3}\, .
    \end{equation}
\end{lemma}
\begin{proof}
   Direct computations show
    \begin{equation}
        \begin{split}
            q^2+\xi^2&= \frac{1}{M_{23}}(m_3 x_2^2+m_2x_3^2-2\sqrt{m_2m_3}x_2x_3)+\frac{1}{M_{23}}(m_2 x_2^2+m_3x_3^2+2\sqrt{m_2m_3}x_2x_3) \\
            &= x_2^2+x_3^2\, .
        \end{split}
    \end{equation}
    Consequently \eqref{invariant_measure} holds.
    
    The set $\partial K_{23}(\gamma)$ is determined by \begin{equation}\label{defining_k_23}
        \abs{\frac{x_2}{\sqrt{m_2}} - \frac{x_3}{\sqrt{m_3}}} = \gamma \abs{x} \, .
    \end{equation}
    Expressing $x_2$ and $x_3$ in $(q,\xi)$--varibales gives
    \begin{equation} \label{x_2x_3_new_coos}
        x_2 = \left( \frac{m_2}{M_{23}} \right)^{1/2} \left( \xi + \left( \frac{m_3}{m_2} \right)^{1/2} q\right), \quad x_3 = \left( \frac{m_3}{M_{23}} \right)^{1/2} \left( \xi - \left( \frac{m_2}{m_3} \right)^{1/2} q\right) \, .
    \end{equation}
    Inserting \eqref{x_2x_3_new_coos} into \eqref{defining_k_23} gives
    \begin{equation} \label{parallel_to_q}
        \abs{\frac{x_2}{\sqrt{m_2}} - \frac{x_3}{\sqrt{m_3}} }  = (\mu_{23})^{-1/2}\abs{q},
    \end{equation}
    for $\mu_{23}$ in \eqref{kappa_0_mu_23}.

    Combining \eqref{defining_k_23} and \eqref{parallel_to_q} shows that the surface $\partial K_{23}(\gamma)$ is determined by the relation
    \begin{equation} \label{points_on_del_K_23}
        \abs{q} = (\mu_{23})^{1/2} \gamma \abs{(x_1,q,\xi)} \, .
    \end{equation}
     Solving  \eqref{points_on_del_K_23} for $\abs{q}$ yields
    \begin{equation} \label{surface_in_q_xi}
        \abs{q} = \kappa_0\abs{(x_1,\xi)} \quad \text{ where } \quad \kappa_0 = \left(\frac{\gamma^2 \mu_{23}}{1-\gamma^2 \mu_{23}} \right)^{1/2}\, .
    \end{equation}
    Using $\eqref{surface_in_q_xi}$ as parametrization of the surface $\partial K_{23}(\gamma)$ the relation for the surface element $d\sigma$ of $\partial K_{23}(\gamma)$ in \eqref{surf_element_23_apppendix} follows from direct computations.
\end{proof}
\section{Auxiliary Estimate used in Lemma~\ref{first_lemma}
}\label{app:C}
\begin{lemma} \label{Lemma_G_derivative}
    Let $\xi \in C^\infty([0,\infty))$ with $\xi(t) = 0$ for $t\leq 1$ and $\xi(t)=1$ for $t \geq 2$, such that $\xi(t)\leq 1$ and $\xi'(t) \leq 2$ for any $t\in [0,\infty)$. For any $\omega,\kappa,\beta >0$ we define
\begin{equation} 
     G(\abs{x})  \coloneqq \frac{\abs{x}^\kappa}{1+\omega \abs{x}^\kappa} \xi(\abs{x}/\beta)\, .
\end{equation}
Then 
    \begin{equation} 
        \abs{\nabla G(x)} \leq  \kappa \abs{x}^{-1} G(x), \,  \text{for } \abs{x}>2\beta 
    \end{equation}
     and for 
    \begin{equation}
        \abs{\nabla G(x)} \leq  \beta^{\kappa -1 } \left( 2^{\kappa+1} +  \kappa 2^{\kappa-1} \right), \, \text{ for } \abs{x} \in [\beta, 2\beta] .
    \end{equation}
\end{lemma}
\begin{proof}
    By construction 
    \begin{equation}
        \xi \left( \frac{\abs{x}}{\beta} \right) \equiv 1
    \end{equation}
    for $\abs{x}\geq 2\beta$ and consequently
    \begin{equation}
        \abs{(\nabla G)(x)} = \abs{ \nabla \left(  \frac{\abs{x}^\kappa}{1+\omega \abs{x}^\kappa} \right)} = \frac{\kappa \abs{x}^{\kappa-1}}{\left( 1+\omega \abs{x}^\kappa\right)^2} \leq \kappa \abs{x}^{-1} G(x) \, .
    \end{equation}
    For $\abs{x} \in [\beta,2\beta]$ 
    \begin{equation}
        \xi \left( \frac{\abs{x}}{\beta} \right) \leq 2 ,
    \end{equation}
    then for $\abs{x} \in [\beta,2\beta]$ 
    \begin{equation}
        \begin{split}
            \abs{(\nabla G)(x)} &\leq \frac{2}{\beta} \frac{\abs{x}^\kappa}{1+\omega \abs{x}^\kappa} + \abs{ \nabla \left(  \frac{\abs{x}^\kappa}{1+\omega \abs{x}^\kappa} \right)} \\
            &\leq 2^{\kappa+1} \beta^{\kappa-1} + \kappa 2^{\kappa-1} \beta^{\kappa -1 } \\
            &=   \beta^{\kappa -1 } \left( 2^{\kappa+1} +  \kappa 2^{\kappa-1} \right) \, .
        \end{split}
    \end{equation}
\end{proof}
\section{Proof of Lemma 6.2 of \cite{BBV:2022} } \label{app:D}
\begin{lemma} \label{1D_hamiltonain_estimate}
    Let $h=-\partial_q^2 + V(q)$  on $L^2(\R)$, such that $h\geq 0$. Asume that for every $\varepsilon>0$ there exists a constant $C(\varepsilon)>0$, such that
    \begin{equation}
        \int_{\R} \abs{V(q)} \abs{\psi(q)}^2 dq \leq \varepsilon  \int_{\R} \abs{\psi'(q)}^2 dq + C(\varepsilon) \int_{\R} \abs{ \psi(q)}^2 dq \quad \forall \psi \in H^1(\R) \, .
    \end{equation}
    Assume there are constants $\tilde C,A,\delta>0$ such that 
    \begin{equation} \label{1d_short_range}
        \abs{V(q)} \leq \tilde C \abs{q}^{-2-\delta}
    \end{equation}
    for any $q\in \R$ with $\abs{q}>A$. Then there exists a constant $C>0$, such that for any $L>0$ and any function $\psi \in H^1(\R)$
    \begin{equation}
        \mathcal{J}(\psi,L) \coloneqq \int_{-L}^L 
 \left( \abs{\psi'(q)}^2  + V(q) \abs{\psi(q)}^2 \right) dq \geq - C L^{-1-\delta}\left( \abs{\psi(L)}^2 + \abs{\psi(-L)}^2\right)
    \end{equation}
\end{lemma}
\begin{proof}
    Let $L>A$, for $s>0$ we define for $\abs{q}\leq L +2$ 
    \begin{equation}
        \psi_s(q) \coloneqq \begin{cases}
            \psi(L)\frac{(L+s)-q}{s}  &q \in (L, (L+s))\\
            \psi(q) &q\in[-L,L] \\
            \psi(-L)\frac{(L+s)+q}{s}  &q \in (-(L+s), -L)
        \end{cases}
    \end{equation}
    and $\psi_s(q) = 0$ for $\abs{q}>L+s$. Then $\psi_s \in H^1(\R)$. By construction $\mathcal{J}(\psi,L) = \mathcal{J}(\psi_s,L) $ and since $h\geq 0$
    \begin{equation} \label{outside_J}
        \begin{split}
            \mathcal{J}(\psi,L)  \geq - &\int_{L}^\infty \left( \abs{\psi_s'(q)}^2  + \abs{V(q)}  \abs{\psi_s(q)}^2 \right) dq \\
            - &\int_{-\infty}^{-L} \left( \abs{\psi_s'(q)}^2  + \abs{V(q)} \abs{\psi_s(q)}^2 \right) dq 
        \end{split}
    \end{equation}
    We estimate the first integral on the right--hand side of \eqref{outside_J}. Using that $\psi'$ and $\psi$ are supported on $\{q\in \R, \abs{q}\leq L+s\}$ yields
    \begin{equation} \label{two_terms_to_be_est}
        \int_{L}^\infty \left( \abs{\psi_s'(q)}^2  + \abs{V(q)}  \abs{\psi_s(q)}^2 \right) dq=\int_{L}^{L+s}  \abs{\psi_s'(q)}^2 dq + \int_{L}^{L+s} \abs{V(q)}  \abs{\psi_s(q)}^2dq \, .
    \end{equation}
    The first term on the right--hand side of \eqref{two_terms_to_be_est} vanishes in the limit $s\to \infty$ and for the second term  we find
    \begin{equation}\label{remaining_integral_withV}
        \int_{L}^{L+s} \abs{V(q)}  \abs{\psi_s(q)}^2dq \leq \abs{\psi(L)}^2 \int_{L}^{L+s}\abs{V(q)} dq
    \end{equation}
    Applying \eqref{1d_short_range} and solving the remaining integral on the right--hand side of \eqref{remaining_integral_withV} we find there exists a constant $C>0$ depending on $\delta>0$ but independent of $L$ and $s$ such that
    \begin{equation} \label{psi_on_edge}
         \int_{L}^{L+s} \abs{V(q)}  \abs{\psi_s(q)}^2dq \leq \frac{C}{L^{1+\delta}}\abs{\psi(L)}^2 \, .
    \end{equation}
    Using the same Arguments for the Integral over $(-\infty,-L]$ in \eqref{outside_J} proves the statement in the limit $s\to \infty$.
\end{proof}
%

\providecommand{\MR}{\relax\ifhmode\unskip\space\fi MR }
\providecommand{\MRhref}[2]{%
  \href{http://www.ams.org/mathscinet-getitem?mr=#1}{#2}
}

\bibliographystyle{annotate}
\bibliography{ref}

\end{document}